\documentclass{article}
\usepackage{pifont,xspace,epsfig,wrapfig,paralist,fullpage}

\usepackage{amsmath,amsthm,amssymb,subfigure,times,algorithmic,url}
\theoremstyle{plain}
\usepackage{latexsym,paralist,multirow,tablefootnote}
\usepackage{graphicx}
\usepackage{amsfonts}
\usepackage{mathtools}
\usepackage{makecell}
\usepackage{caption}
\usepackage[linesnumbered,ruled,lined,boxed]{algorithm2e}
\newcounter{magicrownumbers}
\newcommand\rownumber{\stepcounter{magicrownumbers}\arabic{magicrownumbers}}
\newenvironment{CompactEnumerate}{
\addtolength{\textheight}{.8in}
\addtolength{\textwidth}{.8in}
\addtolength{\oddsidemargin}{-.4in}
\addtolength{\evensidemargin}{-.4in}
\addtolength{\topmargin}{-.2in}

\begin{list}{\arabic{enumi}.}{%
\usecounter{enumi}
\setlength{\leftmargin}{12pt}
\setlength{\itemindent}{3pt}
\setlength{\topsep}{3pt}
\setlength{\itemsep}{1pt}
}}
{\end{list}}
\newenvironment{CompactItemize}{
\begin{list}{$\bullet$}{%
\setlength{\leftmargin}{12pt}
\setlength{\itemindent}{1pt}
\setlength{\topsep}{1pt}
\setlength{\itemsep}{-1pt}
}}
{\end{list}}
\newcommand{\resref}[1]{{\bf (\ref{res:#1})}}
\newtheorem{theorem}{Theorem}[section]
\newtheorem*{non-theorem}{Theorem}
\newtheorem{lemma}[theorem]{Lemma}

\newtheorem{proposition}[theorem]{Proposition}

\newtheorem{remark}[theorem]{Remark}

\newtheorem{corollary}[theorem]{Corollary}
\newtheorem{definition}{Definition}

\def \TreeMech {\textsc{TreeMech}\xspace}
\def \PrivIncERM {\textsc{PrivIncERM}\xspace}
\def \PrivIncReg {\textsc{PrivIncReg$_1$}\xspace}
\def \ProjPrivIncReg {\textsc{PrivIncReg$_2$}\xspace}
\def \NoisyProjGrad {\textsc{NoisyProjGrad}\xspace}
\def \ProjGrad {\textsc{ProjGrad}\xspace}
\def \polylog{\mbox{\rm polylog}}
\def \y {\mathbf y}
\def \eps {\epsilon}

\def \hinge {\mathrm{hinge}}
\def \q {\mathbf q}
\def \z {\mathbf z}

\def \CCC {\mathcal{C}}
\def \EEE {\mathcal{E}}

\def \RRR {\mathcal{R}}
\def \GGG {\mathcal{G}}
\def \XXX {\mathcal{X}}
\def \ZZZ {\mathcal{Z}}
\def \YYY {\mathcal{Y}}
\def \NNN {\mathcal{N}}
\def \LLL {\mathcal{L}}
\def \JJJ {\mathcal{J}}
\def \E {\mathbb{E}}

\def \u {\mathbf u}
\def \N {\mathbb N}

\def \Bin {\mathrm{Bin}}
\def \true {{\mathrm{true}}}
\def \priv {{\mathrm{priv}}}

\def \proj {{\mathrm{proj}}}
\def \OPT {{\mathrm{OPT}}}

\def \poly {\mathrm{poly}}

\def \Alg {\mathrm{Alg}}

\def \a {\mathbf a}
\def \b {\mathbf b}

\def \s {\mathbf s}
\def \g {\mathbf g}

\def \q {\mathbf q}
\def \v {\mathbf v}
\def \x {\mathbf x}
\def \w {\mathbf w}

\def \R {\mathbb{R}}
\def \N {\mathbb{N}}

\newenvironment{mechanism}[1][htb]
{
\begin{algorithm}[#1]%
}{\end{algorithm}}

\begin{document}
\title{Private Incremental Regression}
\author{
Shiva Prasad Kasiviswanathan\thanks{Samsung Research America, \texttt{kasivisw@gmail.com}}
\and Kobbi Nissim\thanks{Georgetown University, Supported by NSF grant CNS1565387 and grants from the Sloan Foundation. \texttt{Kobbi.nissim@georgetown.edu}}
\and
\and Hongxia Jin \thanks{Samsung Research America, \texttt{hongxia.jin@samsung.com}}
}
\date{}
\maketitle
\abstract{
Data is continuously generated by modern data sources, and a recent challenge in machine learning has been to develop techniques that perform well in an incremental (streaming) setting. A variety of offline machine learning tasks are known to be feasible under differential privacy, where generic construction exist that, given a large enough input sample, perform tasks such as PAC learning, Empirical Risk Minimization (ERM), regression, etc. In this paper, we investigate the problem of private machine learning, where as common in practice, the data is not given at once, but rather arrives incrementally over time.

We introduce the problems of {\em private incremental ERM} and {\em private incremental regression} where the general goal is to always maintain a good empirical risk minimizer for the history observed under differential privacy. Our first contribution is a generic transformation of private batch ERM mechanisms into private incremental ERM mechanisms, based on a simple idea of invoking the private batch ERM procedure at some regular time intervals. We take this construction as a baseline for comparison. We then provide two mechanisms for the private incremental regression problem. Our first mechanism is based on privately constructing a noisy incremental gradient function, which is then used in a modified projected gradient procedure at every timestep. This mechanism has an excess empirical risk of $\approx\sqrt{d}$, where $d$ is the dimensionality of the data. While from the results of Bassily~\emph{et al.}\ \cite{bassily2014differentially} this bound is tight in the worst-case, we show that certain geometric properties of the input and constraint set can be used to derive significantly better results for certain interesting regression problems. Our second mechanism which achieves this is based on the idea of projecting the data to a lower dimensional space using random projections, and then adding privacy noise in this low dimensional space. The mechanism overcomes the issues of adaptivity inherent with the use of random projections in online streams, and uses recent developments in high-dimensional estimation to achieve an excess empirical risk bound of $\approx T^{1/3} W^{2/3}$, where $T$ is the length of the stream and $W$ is the sum of the Gaussian widths of the input domain and the constraint set that we optimize over.
}

\section{Introduction}
Most modern data such as documents, images, social media data, sensor data, and mobile data naturally arrive in a streaming fashion, giving rise to the challenge of incremental machine learning, where the goal is build and publish a model that evolves as data arrives. Learning algorithms are frequently run on sensitive data, such as location information in a mobile setting, and results of such analyses could leak sensitive information. For example, Kasiviswanathan~\emph{et al.}\ \cite{kasiviswanathan2013power} show how the results of many convex ERM problems can be combined to carry out reconstruction attacks in the spirit of Dinur and Nissim~\cite{DiNi03}. Given this, a natural direction to explore, is whether, we can carry out incremental machine learning, without leaking any significant information about individual entries in the data. For example, a data scientist, might want to continuously update the regression parameter of a linear model built on a stream of user profile data gathered from an ongoing survey, but these updates should not reveal whether any one person participated in the survey or not. 

Differential privacy~\cite{DMNS06} is a rigorous notion of privacy that is now widely studied in computer science and statistics. Intuitively, differential privacy requires that datasets differing in only one entry induce similar distributions on the output of a (randomized) algorithm. One of the strengths of differential privacy comes from the large variety of machine learning tasks that it allows. Good generic constructions exist for tasks such as PAC learning~\cite{BDMN05,KLNRS08} and Empirical Risk Minimization~\cite{rubinstein2009learning,kifer2012private,DBLP:journals/jmlr/JainKT12,jain2013differentially,thakurta2013differentially,jain2014near,bassily2014differentially,duchi2013local,ullman2015private,talwar2014private}. These constructions, however, are typically focused on  the batch (offline) setting, where information is first collected and then analyzed. Considering an incremental setting, it is natural to ask whether these tasks can still be performed with high accuracy, under differential privacy. 

In this paper, we introduce the problem of private incremental empirical risk minimization (ERM) and provide algorithms for this new setting. Our particular focus will be on the problem of private incremental linear regression. Let us start with a description of the traditional batch convex ERM framework. Given a dataset and a constraint space $\CCC$, the goal in ERM is to pick a $\theta \in \CCC$ that minimizes the {\em empirical error (risk)}. Formally, given $n$ datapoints $\z_1,\dots,\z_n$ from some domain $\ZZZ$, and a closed, convex
set $\CCC \subseteq \R^d$, consider the optimization problem:
\begin{align} \label{eqn:mestimator}
 \min_{\theta \in \CCC}\, \JJJ(\theta;\z_1,\dots,\z_n) \mbox{ where } \JJJ(\theta;\z_1,\dots,\z_n)  =  \sum_{i=1}^n \jmath(\theta;\z_i).
\end{align}
The loss function $\JJJ : \CCC \times \ZZZ^n \rightarrow \R$ measures the fit of $\theta \in \CCC$ to the given data $\z_1,\dots,\z_n$, the function $\jmath : \CCC \times \ZZZ  \rightarrow \R$ is the loss associated with a single datapoint and is assumed to be convex in the first parameter $\theta$ for every $\z \in \ZZZ$. It is common to assume that the loss function has certain properties, e.g., positive valued. The {\em M-estimator} (true empirical risk minimizer) $\hat{\theta}$ associated with a given a function $\JJJ(\theta; \z_1,\dots,\z_n) \geq 0$ is defined as:\footnote{This formulation also captures {\em regularized} ERM, in which an additional convex function $R(\theta)$ is added to the loss function to penalize certain type of solutions, e.g., ``penalize'' for the ``complexity'' of $\theta$. The loss function $\JJJ(\theta; \z_1,\dots,\z_n)$ then equals $\sum_{i=1}^n \jmath (\theta; \z_i) + R(\theta)$, which is same as replacing $\jmath (\theta; \z_i)$ by $\jmath (\theta; \z_i) + R(\theta)/n$ in~\eqref{eqn:mestimator}.}
$$\hat{\theta} \in \mbox{argmin}_{\theta \in \CCC}\, \JJJ(\theta; \z_1,\dots,\z_n) = \mbox{argmin}_{\theta \in \CCC}\, \sum_{i=1}^n  \jmath (\theta; \z_i).$$
This type of program captures a variety of empirical risk minimization (ERM) problems, e.g., the MLE (Maximum Likelihood Estimators) for linear regression is captured by setting $\jmath(\theta;\z) = (y - \langle \x,\theta \rangle)^2$ in~\eqref{eqn:mestimator}, where $\z = (\x,y)$ for $\x \in \R^d$ and $y \in \R$. Similarly, the MLE for logistic regression is captured by setting $\jmath(\theta;\z) = \ln(1+\exp(-y\langle \x,\theta\rangle))$. Another common example is the support vector machine (SVM), where $\jmath(\theta;\z) = \hinge(y \langle \x,\theta \rangle)$, where $\hinge(a) = 1-a$ if $a \leq 1$ and $0$ otherwise. 

The main focus of this paper will be on a particularly important ERM problem of linear regression. Linear regression is a popular statistical technique that is commonly used to model the relationship between the outcome (label) and the explanatory variables (covariates). Informally, in a linear regression, given $n$ covariate-response pairs $(\x_1,y_1),\dots,(\x_n,y_n) \in \R^d \times \R$, we wish to find a (regression) parameter vector $\hat{\theta}$ such that $\langle \x_i, \hat{\theta} \rangle \approx y_i$ for most $i$'s. Specifically, let $\y = (y_1,\dots,y_n) \in \R^n$ denote a vector of the responses, and let $X \in \R^{n \times d}$ be the design matrix where $\x_i^\top$ is the $i$th row. Consider the linear model: $\y = X\theta^\star+ \w$, where $\w$ is the noise vector, the goal in linear regression is to estimate the unknown regression vector $\theta^\star$. Assuming that the noise vector $\w=(w_1,\dots,w_n)$ follows a (sub)Gaussian distribution, estimating the vector $\theta^\star$ amounts to solving the ``ordinary least squares'' (OLS) problem: $$\hat{\theta} \in \mbox{argmin}_\theta\, \sum_{i=1}^n (y_i - \langle \x_i,\theta \rangle)^2.$$ 
Typically, for additional guarantees such as sparsity, stability, etc., $\theta$ is constrained to be from a convex set $\CCC \subset \R^d$. Popular choices of $\CCC$ include the $L_1$-ball (referred to as {\em Lasso} regression) and $L_2$-ball (referred to as {\em Ridge} regression). In an incremental setting, the $(\x_i,y_i)$'s arrive over time, and the goal in incremental linear regression is to maintain over time (an estimate of) the regression parameter. We provide a more detailed background on linear regression in Appendix~\ref{app:linearregressionbackground}.

\noindent\textbf{Incremental Setting.} 
In an incremental setting the data arrives in a stream at discrete time intervals. The incremental setting is a variant of the traditional batch setting capturing the fact that modern data is rarely collected at one single time and more commonly data gathering and analysis may be interleaved. An incremental algorithm is modeled as follows: at each timestep the algorithm receives an input from the stream, computes, and produces outputs.  Typically, constraints are placed on the algorithm in terms of some computational resource (such as memory, computation time) availability. In this paper, the challenge in this setting comes from the differential privacy constraint because frequent releases about the data can lead to privacy loss (see, e.g.,~\cite{DiNi03,kasiviswanathan2013power}).

We focus on the incremental ERM problem.  In this problem setting, the data $\z_1,\dots,\z_T \in \ZZZ$ arrives one point at each timestep $t \in \{1,\dots,T\}$. The goal of an incremental ERM algorithm is to release at each timestep $t$, an estimator that minimizes the risk measured on the data $\z_1,\dots,\z_t$. In more concrete terms, the goal is to output $\hat{\theta}_t$ at every $t \in \{1,\dots,T\}$, where: 
\begin{align*}
\mbox{Incremental ERM:} \quad \hat{\theta}_t \in \mbox{argmin}_{\theta \in \CCC}\,\sum_{i=1}^t \jmath(\theta;\z_i).
\end{align*}
The goal of {\em private incremental ERM} algorithm, is to (differential) privately estimate $\hat{\theta}_t$ at every $t \in \{1,\dots,T\}$.\!\footnote{A point to note in the above description, while we have the privacy constraint, we have placed no computational constraints on the algorithm. In particular, the above description allows also those algorithms that at time $t$ use the whole input history $\z_1,\dots,\z_t$. However, as we will discuss in Sections~\ref{sec:smalld} and~\ref{sec:width}, our proposed approaches for private incremental regression are also efficient in terms of their resource requirements.} In this paper, we develop the first private incremental ERM algorithms.

There is a long line of work in designing differentially private algorithms for empirical risk minimization problems in the batch setting~\cite{rubinstein2009learning,kifer2012private,jain2013differentially,thakurta2013differentially,jain2014near,bassily2014differentially,duchi2013local,ullman2015private,talwar2014private}. A naive approach to transform the existing batch techniques to work in the incremental model is by using them to recompute the outcome after each datapoint's arrival. However, for achieving an overall fixed level of differential privacy, this would result in an unsatisfactory loss in terms of utility. Precise statements can be obtained using composition properties of differential privacy (as in Theorem~\ref{thm:composition2}), but informally if a differentially private algorithm is executed $T$ times on the same input, then the privacy parameter ($\eps$ in Definition~\ref{defn:dp}) degrades by a factor of $\approx\sqrt{T}$, and this affects the overall utility of this approach as the utility bounds typically depend inversely on the privacy parameter. Therefore, the above naive approach will suffer an additional multiplicative factor of $\approx \sqrt{T}$ over the risk bounds obtained for the batch case.
Our goal is to obtain risk bounds in the incremental setting that are comparable to those in the batch setting, i.e., bounds that do not suffer from this additional penalty of $\sqrt{T}$.


\smallskip
\noindent\textbf{Generalization.} For the kind of problems we consider in this paper, if the process generating the data satisfies the conditions for {\em generalization} (e.g., if the data stream contains datapoints where each datapoint is sampled independent and identically from an unknown distribution), the incremental ERM would converge to the true risk on the distribution (via uniform convergence and other ideas, refer to~\cite{vapnik2013nature,shalev2010learnability,bassily2014differentially} for more details). In this case, the model learned in an incremental fashion will have a good predictive accuracy on unseen arriving data. If however, the data does not satisfy these conditions, then $\hat{\theta}_t$ can be viewed as a ``summarizer'' for the data seen so far. Generating these summaries could also be useful in many applications, e.g., the regression parameter can be used to explain the associations between the outcome ($y_i$'s) and the covariates ($\x_i$'s). These associations are regularly used in domains such as public health, social sciences, biological studies, to understand whether specific variables are important (e.g., statistically significant) or unimportant predictors of the outcome. In practice, these associations would need to be constantly re-evaluated over time as new data arrives. 

\smallskip
\noindent\textbf{Comparison to Online Learning.} 
Online learning (or sequential prediction) is another well-studied setting for learning when the data arrives in a sequential order. There are differences between the goals of incremental and online learning.  In online ERM learning, the aim is to provide an estimator $\tilde{\theta}_t$ that can be used for future prediction. More concretely, at time $t$, an online learner, chooses $\tilde{\theta}_t$ and then the adversary picks $\z_t$ and the learner suffers loss of $\jmath(\tilde{\theta}_t;\z_t)$. Online learners try to minimize {\em regret} defined as the difference between the cumulative loss of the learner and the cumulative loss of the best fixed decision at the end of $T$ rounds~\cite{shalev2011online}. In an incremental setting, the algorithm gets to observe $\z_t$ before committing to the estimator, and the goal is to ensure at each timestep $t$, the algorithm maintains a {\em single} estimator that minimizes the risk on the history. There are strong lower bounds on the achievable regret for online ERM learning. In particular, even under the differential privacy constraint, the excess risk bounds on incremental learning that we obtain here have better dependence on $T$ (stream length) than the regret lower bounds for online ERM.
Incremental learning model should be viewed as a variant of batch (offline) learning model, where the data arrives over time, and the goal is to output intermediate results.

\subsection{Our Results}
Before stating our results, let us define how we measure success of our algorithms. As is standard in ERM, we measure the quality of our algorithm by the worst-case (over inputs) excess (empirical) risk (defined as the difference from the minimum possible risk over a function class). In an incremental setting, we want this excess risk to be always small for any sequence of inputs. The following definition captures this requirement. All our algorithms are randomized, and they take a confidence parameter $\beta$ and produce bounds that hold with probability at least $1-\beta$. 
\begin{definition} \label{defn:utilerm}
A (randomized) incremental algorithm is an $(\alpha,\beta)$-{\em estimator} for loss function $\JJJ$, if with probability at least $1 -\beta$ over the coin flips of the algorithm, for each $t \in \{1,\dots,T\}$, after processing a prefix of the stream of length $t$, it generates an output $\theta_t \in \CCC$ that satisfies the following bound on excess (empirical) risk:
$$\JJJ(\theta_t;\z_1,\dots,\z_t)- \JJJ(\hat{\theta}_t;\z_1,\dots,\z_t) \leq \alpha,$$
where $\hat{\theta}_t \in \mbox{argmin}_{\theta \in \CCC}\, \JJJ(\theta;\z_1,\dots,\z_t)$ is the true empirical risk minimizer. Here, $\alpha$ is referred to as the bound on the excess risk.
\end{definition}
\begin{table*}[tb]
\begin{small}
\centering
\begin{tabular} {|c|c|c|}
\hline 
& \makecell{Incremental ERM \\ Problem Objective}  & \makecell{Bound on the Excess (Empirical) Risk under $(\eps,\delta)$-Differential Privacy \\ ($\alpha$ in Definition~\ref{defn:utilerm})} \\
\hline \hline
\rownumber & \makecell{Convex Function\\ (using a generic transformation)}   &  $\frac{(Td)^{\frac 1 3} \log^{\frac 5 2}(1/\delta)}{\eps^{2/3}}$ \\ 
\hline
\rownumber & \makecell{Strongly Convex Function\\ (using a generic transformation)} & $\frac{\sqrt{d}\log^4(1/\delta)}{\nu^{1/2}\eps}$\\ 
\hline \hline
\rownumber & \makecell{Linear Regression} & \begin{tabular}{c}Mech 1: $\frac{\sqrt{d}\sqrt{\log(1/\delta)}}{\eps}$ \\ \hline \makecell{Mech 2: $\frac{T^{\frac 1 3}W^{\frac 2 3}\sqrt{\log(1/\delta)}}{\eps} + T^{\frac 1 6}W^{\frac 1 3}\sqrt{\OPT}+T^{\frac 1 4}W^{\frac 1 2}\sqrt[4]{\OPT}$ \\(\begin{small}where $W=w(\XXX)+w(\CCC)$ and $\OPT$ is the minimum empirical risk at timestep $T$) \end{small}} \end{tabular}\\ \hline
\end{tabular}
\end{small}
\caption{Summary of our results. The stream length is $T$ and $d$ is the dimensionality of the data (number of covariates in the case of regression). The bounds are stated for the setting where both the Lipschitz constant of the function and $\| \CCC \|$ are scaled to $1$. The bounds ignore polylog factors in $d$ and $T$, and the value in the table gives the bound when it is below $T$, i.e., the bounds should be read as $\min\{T,\cdot\}$. $\nu$ is the strong convexity parameter. For the regression problem, $\XXX$ is the domain from which the inputs (covariates) are drawn and $\CCC$ is the constraint space. $\OPT$ stands for the minimum empirical risk at time $T$. The exact results are provided in Theorems~\ref{thm:thm1},~\ref{thm:thm2}, and~\ref{thm:thm6}, respectively.
}
\label{tab:results}
\end{table*}

In this paper, we propose incremental ERM algorithms that provide a differential privacy guarantee on data streams. Informally, differential privacy requires that the output of a data analysis mechanism is not overly affected by any single entry in the input dataset. In the case of incremental setting, we insist that this guarantee holds at each timestep for all outputs produced up to that timestep (precise definition in Section~\ref{sec:prelim}). This would imply that, an adversary is unable to determine whether a particular datapoint was present or not in the input stream by observing the output of the algorithm over time. Two parameters, $\eps$ and $\delta$ control the level of privacy. Very roughly, $\eps$ is an upper bound on the amount of influence an individual datapoint has on the outcome and $\delta$ is the probability that this bound fails to hold (for a precise interpretation, refer to~\cite{KS08}), so the definition becomes more stringent as $\eps,\delta \rightarrow 0$. Therefore, while parameters $\alpha,\beta$ measure the accuracy of an incremental algorithm, the parameters $\eps,\delta$ represent its privacy risk. Our private incremental algorithms take $\eps,\delta$ as parameters and satisfies: a) the differential privacy constraint with parameters $\eps,\delta$ and b) $(\alpha,\beta)$-estimator property (Definition~\ref{defn:utilerm}) for every $\beta > 0$ and some $\alpha$ (parameter that the algorithm tries to minimize).

There is a trivial differentially private incremental ERM algorithm that completely ignores the input and outputs at every $t \in \{1,\dots,T\}$, any $\theta \in \CCC$ (this scheme is private as the output is always independent of the input). The excess risk of this algorithm is at most $2 T L \| \CCC \|$,\!\footnote{This follows as in each timestep $t \in \{1,\dots,T\}$, for any $\theta \in \CCC$, $\JJJ(\theta;\z_1,\dots,\z_t)- \JJJ(\hat{\theta}_t;\z_1,\dots,\z_t) \leq t L \| \theta - \hat{\theta}_t \| \leq t L (\|\theta \| + \|\hat{\theta}_t \|) \leq  2 t L \| \CCC\| \leq 2 T L \| \CCC\|$.} where $L$ is the Lipschitz constant of $\jmath$ (Definition~\ref{defn:Lipsfunc}) and $\| \CCC \|$ is the maximum attained norm in the convex set $\CCC$ (Definition~\ref{defn:diameter}). All bounds presented in this paper, as also true for all other existing results in the private ERM literature, are only interesting in the regime where they are less than this trivial bound.

For the purposes of this section, we make some simplifying assumptions and omit\footnote{This includes parameters $\eps,\delta, \beta, \|\CCC\|$, and the Lipschitz constant of the loss function.} dependence  on all but key variables (dimension $d$ and stream length $T$). Slightly more detailed bounds are stated in Table~\ref{tab:results}. All our algorithms run in time polynomial in $d$ and $T$ (exact bounds depend on the time needed for Euclidean projection onto the constraint set which will differ based on the constraint set). Additionally, our private incremental regression algorithms, which utilize the \emph{Tree Mechanism} of~\cite{DNPR10,CSS11} have space requirement whose dependence on the stream length $T$ is only logarithmic.

\begin{list}{{\bf (\arabic{enumi})}}{\usecounter{enumi}
\setlength{\leftmargin}{12pt}
\setlength{\listparindent}{\parindent}
\setlength{\parsep}{0pt}}
\item \label{res:con1} \textbf{A Generic Transformation.} A natural first question is whether {\em non-trivial} private ERM algorithms exist in general for the incremental setting. Our first contribution is to answer this question in the affirmative -- we present a simple generic transformation of private batch ERM algorithms into private incremental ERM algorithms. The construction idea is simple: rather than invoking the batch ERM algorithm every timestep, the batch ERM algorithm is invoked every $\tau$ timesteps, where $\tau$ is chosen to balance the excess risk factor coming from the stale risk minimizer (because of inaction) with the excess risk factor coming from the increased privacy noise due to reuse of the data. Using this idea along with recent results of Bassily~\emph{et al.}\ \cite{bassily2014differentially} for private batch ERM, we obtain an excess risk bound ($\alpha$ in Definition~\ref{defn:utilerm}) of $\tilde{O}(\min\{(Td)^{1/3},T\})$.\footnote{For simplicity of exposition, the $\tilde{O}(\cdot)$ notation hides factors polynomial in $\log T$ and $\log d$.} Using this same framework, we also show that when the loss function is strongly convex (Definition~\ref{defn:strongconvex}) the excess risk bound can be improved to $\tilde{O}(\min\{\sqrt{d},T\})$. 

A follow-up question is: how much worse is this private incremental ERM risk bound compared to best known private batch ERM risk bound (with a sample size of $T$ datapoints)? In the batch setting, the results of Bassily~\emph{et al.}\ \cite{bassily2014differentially} establish that for any convex ERM problem, it is possible to achieve an excess risk bound of $\tilde{O}(\min\{\sqrt{d},T\})$ (which is also tight in general). Therefore, our transformation from the batch to incremental setting, causes the excess risk to increase by at most a factor of $\approx \max\{T^{1/3}/d^{1/6},1\}$. Note that even for a low-dimensional setting (small $d$ case), the factor increase in excess risk of $\approx T^{1/3}$  (as now $\max\{T^{1/3}/d^{1/6},1\} \approx T^{1/3}$) is much smaller than the factor increase of $\approx T^{1/2}$ for the earlier described naive approach based on using a private batch ERM algorithm at every timestep. The situation only improves for larger dimensional data.


\item \label{res:con2} \textbf{Private Incremental Linear Regression Using Tree Mechanism.} 
We show that we can improve the generic construction from~\resref{con1} for the important problem of linear regression. We do so by introducing the notion of a {\em private gradient function} (Definition~\ref{defn:privgradfunc}) that allows differentially private evaluation of the gradient at any $\theta \in \CCC$.\!\footnote{Intuitively, this would imply that for any $\theta \in \CCC$, the output of a private gradient function cannot be used to distinguish two streams that are almost the same (differential privacy guarantee).} More formally, for any $\theta \in \CCC$, a private gradient function allows differentially private evaluation of the gradient at $\theta$ to within a small error (with high probability). Since the data arrives continually, our algorithm utilizes the \emph{Tree Mechanism} of~\cite{DNPR10,CSS11} to continually update the private gradient function. Now given access to a private gradient function, we can use any first-order convex optimization technique (such as projected gradient descent) to privately estimate the regression parameter. The idea is, since these optimization techniques operate by iteratively taking steps in the direction opposite to the gradient evaluated at the current point, we can use a private gradient function for evaluating all the required gradients. Using this, we design a private regression algorithm for the incremental setting, that achieves an excess risk bound of $\tilde{O}(\min\{\sqrt{d},T\})$. It is easy to observe, for private incremental linear regression, this result improves the bounds from the generic construction above for any choice of $d,T$ (as $\min\{\sqrt{d},T\} \leq \min\{(Td)^{1/3},T\}$).\!\footnote{A linear regression instance typically does not satisfy strong convexity requirements.} 

Ignoring $\polylog$ factors, this worst-case bound matches the lower bounds on the excess risk for squared-loss in the batch case~\cite{bassily2014differentially}, implying that this bound cannot be improved in general (an excess risk upper bound for a problem in the incremental setting trivially holds for the same problem in the batch setting).

\item \label{res:con3} \textbf{Private Incremental Linear Regression: Going Beyond Worst-Case.} 
The noise added in our previous solution \resref{con2} grows approximately as the square root of the input dimension (for sufficiently large $T$), which could be prohibitive for a high-dimensional input. While in a worst-case setup this, as we discussed above, seems unavoidable, we investigate whether certain geometric properties of the input/output space could be used to obtain better results in certain interesting scenarios.

A natural strategy for reducing the dependence of the excess risk bounds on $d$ is to use dimensionality reduction techniques such as the Johnson-Lindenstrauss Transform (JLT): The server performing the computation can choose a random projection matrix $\Phi \in \R^{m \times d}$ which is then used for projecting all the $\x_i$'s (covariates) onto a lower-dimensional space. The advantage being that, using the techniques from~\resref{con2}, one could privately minimize the excess risk in the projected subspace, and in doing so, the dependence on the dimension in excess risk is reduced to $\approx \sqrt{m}$ (from $\approx \sqrt{d}$).\footnote{A point to note that this use of random projections for linear regression is different from its typical use in prior work~\cite{TCS-060}, where they are used not for reducing dimensionality but rather for reducing the number of samples used in the regression computation, and hence improving in running time.} However, there are two significant challenges in implementing this idea for incremental regression, both of which we overcome using geometric ideas. 

The first one being that this only solves the problem in the projected subspace, whereas our goal is to produce an estimate to true empirical risk minimizer. To achieve this we would need to ``lift'' back the solution from the projected subspace to the original space. We do so by using recent developments in the problem of high dimensional estimation with constraints~\cite{vershynin2014estimation}. The {\em Gaussian width} of the constraint space $\CCC$ plays an important role in this analysis. Gaussian width is a well-studied quantity in convex geometry that captures the geometric complexity of any set of vectors.\footnote{For a set $S \subseteq \R^d$, Gaussian width is defined as $w(S)=\E_{\g \in \NNN(0,1)^d}\, [\sup_{\a \in S} \langle \a, \g\rangle]$.} The rough idea here being that a good estimation (lifting) can be done from few observations (small $m$) as long as the constraint set $\CCC$ has small Gaussian width. Many popular sets have low Gaussian width, e.g., the width of $L_1$-ball ($B_1^d$ is $\Theta(\sqrt{\log d})$, and that of set of all sparse vectors in $\R^d$ with at most $k$ non-zero entries is $\Theta(\sqrt{k \log (d/k)})$.

The second challenge comes from the incremental nature of input generation, because it allows for generation of $\x_i$'s after $\Phi$ is {\em fixed}. This is an issue because the guarantees of random projections, such as JL transform, only hold if the inputs on which it is applied are chosen before choosing the transformation. For example, given a random projection matrix $\Phi \in \R^{m \times d}$ with $m \ll d$, it is simple to generate $\x$ such that the norm on $\x$ is substantially different from the norm of $\Phi \x$.\!\footnote{Note that this issue arises independent of the differential privacy requirements, and holds even in a non-private incremental setting.} Again to deal with these kind of adaptive choice of inputs, we rely on the geometric properties of problem. In particular, we use Gordon's theorem that states that one can embed a set of points $S$ on a unit sphere into a (much) lower-dimensional space $\R^m$ using a Gaussian random matrix\footnote{Recent results~\cite{bourgain2015toward} have shown other distributions for generating $\Phi$ also provide similar guarantees.} $\Phi$ such that, $\sup_{\a \in S} | \| \Phi \a\|^2 - \| \a \|^2 |$ is small (with high probability), provided $m$ is at least the square of the Gaussian width of $S$~\cite{gordon1988milman}. In a sense, $w(S)^2$ can be thought as the ``effective dimension'' of $S$, so projecting the data onto $m \approx w(S)^2$ dimensional space suffices for guaranteeing the above condition. 

Using the above geometric ideas, and the \emph{Tree Mechanism} to incrementally construct the private gradient function (as in~\resref{con2}), we present our second private algorithm for incremental regression, with an 
$$\tilde{O}(\min\{T^{\frac 1 3}W^{\frac 2 3} + T^{\frac 1 6}W^{\frac 1 3}\sqrt{\OPT}+T^{\frac 1 4}W^{\frac 1 2}\sqrt[4]{\OPT},T\})$$ 
excess risk bound, where $W=w(\XXX)+w(\CCC)$, $\XXX \subset \R^d$ is the domain from which the $\x_i$'s (covariates) are drawn, and $\OPT$ is the true minimum empirical risk at time $T$. As we discuss in Section~\ref{sec:width}, for many practically interesting regression instances, such as when $\XXX$ is a domain of sparse vectors, and $\CCC$ is bounded by an $L_1$-ball (as is the case of popular Lasso regression) or is a polytope defined by polynomial (in dimension) many vertices, $W \approx \polylog(d)$, in which case the risk bound can be simplified to $\tilde{O}(\min\{T^{\frac 1 3} + T^{\frac 1 6}\sqrt{\OPT}+T^{\frac 1 4}\sqrt[4]{\OPT},T\})$. In most practical scenarios, when the relationship between covariates and response satisfy a linear relationship, one would also expect $\OPT \ll T$. 
These bounds show that, for certain instances, it is possible to design differentially private risk minimizers in the incremental setting with excess risk that depends only poly-logarithmically on the dimensionality of the data, a desired feature in a high-dimensional setting.
\end{list}

\noindent\textbf{Organization.} In Section~\ref{sec:related}, we discuss some related work. In Section~\ref{sec:prelim}, we present some preliminaries. Our generic transformation from a private batch ERM algorithm to a private incremental ERM algorithm is given in Section~\ref{sec:generic}. We present techniques that improve upon these bounds for the problem of private incremental regression in Sections~\ref{sec:smalld} and~\ref{sec:width}. The appendices contain some proof details and supplementary material. In Appendix~\ref{app:noisyproj}, we analyze the convergence rate of a noisy projected gradient descent technique, and in Appendix~\ref{app:treemech}, we present the~\emph{Tree Mechanism} of~\cite{DNPR10,CSS11}.

\subsection{Related Work} \label{sec:related}
\noindent\textbf{Private ERM.} Starting from the works of Chaudhuri~\emph{et al.}\ \cite{chaudhuri2009privacy,chaudhuri2011differentially}, private convex ERM problems have been studied in various settings including the low-dimensional setting~\cite{rubinstein2009learning,kifer2012private}, high-dimensional sparse regression setting~\cite{kifer2012private,thakurta2013differentially}, online learning setting~\cite{DBLP:journals/jmlr/JainKT12,thakurta2013nearly,jain2014near,DBLP:conf/uai/MishraT15}, local privacy setting~\cite{duchi2013local}, and interactive setting~\cite{jain2013differentially,ullman2015private}.

Bassily~\emph{et al.} \cite{bassily2014differentially} presented algorithms that for a general convex loss function $\jmath(\theta;\z)$ which is $1$-Lipschitz (Definition~\ref{defn:Lipsfunc}) for every $\z$, achieves an expected excess risk of $\approx \sqrt{d}$ under $(\eps,\delta)$-differential privacy and $\approx d$ under $\eps$-differential privacy (ignoring the dependence on other parameters for simplicity).\!\footnote{Better risk bounds were achieved under strong convexity assumptions
~\cite{bassily2014differentially,talwar2014private}.} We use their batch mechanisms in our generic construction to obtain risk bounds for incremental ERM problems (Theorem~\ref{thm:thm1}). They also showed that these bounds cannot be improved {\em in general}, even for the least-squares regression function.
However, if the constraint space has low Gaussian width (such as with the $L_1$-ball), Talwar~\emph{et al.}\ \cite{talwar2014private,talwar2015nearly} recently showed that, under $(\eps,\delta)$-differential privacy, the above bound can be improved by exploiting the geometric properties of the constraint space. An analogous result under $\eps$-differential privacy for the class of {\em generalized linear functions} (which includes linear, logistic regression) was recently obtained by Kasiviswanathan and Jin~\cite{kasiviswanathanicml}. Our excess risk bounds based on Gaussian width (presented in Section~\ref{sec:width}) uses a lifting procedure similar to that used by Kasiviswanathan and Jin~\cite{kasiviswanathanicml}. All of these above algorithms operate in the batch setting, and ours is the first work dealing with private ERM problems in an incremental setting.

\smallskip
\noindent\textbf{Private Online Convex Optimization.} Differentially private algorithms have also been designed for a large class of online (convex optimization) learning problems, in both the full information and bandit settings~\cite{DBLP:journals/jmlr/JainKT12,thakurta2013nearly,jain2014near}. Adapting the popular {\em Follow the Approximate Leader} framework~\cite{hazan2006logarithmic}, Smith and Thakurta~\cite{thakurta2013nearly} obtained regret bounds for private online learning with nearly optimal dependence on $T$ (though the dependence on the dimensionality $d$ in these results is much worse than the known lower bounds). As discussed earlier, incremental learning is a variant of the batch learning, with goals different from online learning.

\smallskip
\noindent\textbf{Private Incremental Algorithms.} Dwork~\emph{et al.}\ \cite{DNPR10} introduced the problem of counting under incremental (continual) observations. The goal is to monitor a stream of $T$ bits, and continually release a counter of the number of $1$'s that have been observed so far, under differential privacy. The elegant \emph{Tree Mechanism} introduced by Dwork~\emph{et al.}\ \cite{DNPR10} and Chan~\emph{et al.}\ \cite{CSS11} solves this problem, under $\eps$-differential privacy, with error roughly $\log^{5/2} T$. The versatility of this mechanism has been utilized in different ways in subsequent works~\cite{thakurta2013nearly,DNRR15,hsu2014private}. We use the \emph{Tree Mechanism} as a basic building block for computing the private gradient function incrementally. Dwork~\emph{et al.}\ \cite{DNPR10} also achieve {\em pan-privacy} for their continual release (which means that the mechanism preserves differential privacy even when an adversary can observe snapshots of the mechanism's internal states), a property that we do not investigate in this paper.


\smallskip
\noindent\textbf{Use of JL Transform for Privacy.}
The use of JL transform for achieving differential privacy with better utility has been well documented for a variety of computational tasks~\cite{zhou2009differential,blocki2012johnson,kenthapadi2013privacy,sheffet2015private,upadhyay2014randomness}. Blocki~\emph{et al.}\ \cite{blocki2012johnson} have shown that if $\Phi \in \R^{m \times n}$ is a Gaussian random matrix of appropriate dimension, then $\Phi X \in \R^{m \times d}$ is differentially private  if the least singular value of the matrix $X \in \R^{n \times d}$ is ``sufficiently'' large. The bound on the least singular value was recently improved by Sheffet~\cite{sheffet2015private}. Here the privacy comes as a result of randomization inherent in the transform. However, these results require that the projection matrix is kept private, which is an issue in an incremental setting, where an adversary could learn about $\Phi$ over time. Kenthapadi~\emph{et al.}\ \cite{kenthapadi2013privacy} use Johnson-Lindenstrauss transform to publish a private sketch that enables estimation of the distance between users. Their main idea is based on projecting a $d$-dimensional user feature vector into a lower $m$-dimensional space by first applying a random Johnson-Lindenstrauss transform and then adding Gaussian noise to each entry of the resulting vector. None of these above results deal with an incremental setting, where applying the JL transform is itself a challenge because of the adaptivity issues.

\smallskip
\noindent\textbf{Traditional Streaming Algorithms.}
The literature on streaming algorithms is replete with various techniques that can solve linear regression and related problems on various streaming models of computation under various computational resource constraints. We refer the reader to the survey by Woodruff~\cite{TCS-060} for more details. However, incremental regression under differential privacy, poses a different challenge than that faced by traditional streaming algorithms. The issue is that the solution (regression parameter) at each timestep relies on all the datapoints observed in the past, and frequent releases about earlier points can lead to privacy loss. 

\section{Preliminaries} \label{sec:prelim}
\noindent\textbf{Notation and Data Normalization.} We denote $[n]=\{1,\ldots,n\}$. Vectors are in column-wise fashion, denoted by boldface letters. For a vector $\v$, $\v^\top$ denotes its transpose, $\|\v\|$ it's Euclidean ($L_2$-) norm, and $\| \v \|_1$ it's $L_1$-norm. For a matrix $M$, $\|M\|$ denotes its spectral norm which equals its largest singular value, and $\| M \|_F$ its Frobenius norm. We use $\mathbf{0}$ to denote a $d$-dimensional vector  of all zeros.  The $d$-dimensional unit ball in $L_p$-norm centered at origin is denoted by $B_p^d$. $\mathbb{I}_d$ represents the $d \times d$ identity matrix. $\NNN(\mu,\Sigma)$ denotes the Gaussian distribution with mean vector $\mu$ and covariance matrix $\Sigma$. For a variable $n$, we use $\poly(n)$ to denote a polynomial function of $n$ and $\polylog(n)$ to denote $\poly(\log(n))$.

We assume all streams are of a fixed length $T$, which is known to the algorithm. We make this assumption for simplifying the discussion. In fact, in our presented generic transformation for incremental ERM this assumption can be straightforwardly removed. Whereas in the case of algorithms for private incremental regression this assumption can be removed by using a simple trick\footnote{Chan \emph{et al.}\ \cite{CSS11} presented a scheme that provides a generic way for converting the \emph{Tree Mechanism} that requires prior knowledge of $T$ into a mechanism that does not.  They also showed that this new mechanism (referred to as the \emph{Hybrid Mechanism}) achieves asymptotically the same error guarantees as the \emph{Tree Mechanism}. The same ideas work in our case too, and the asymptotic excess risk bounds are not affected.} introduced by Chan~\emph{et al.}~\cite{CSS11}. For a stream $\Gamma$, we use $\Gamma_t$ to denote the stream prefix of length~$t$. 

Throughout this paper, we use $\ell$ and $\LLL$ to indicate the least-squared loss on a single datapoint and a collection of datapoints, respectively. Namely,
\begin{eqnarray*}
& \ell(\theta;(\x,y)) = (y - \langle \x,\theta \rangle)^2 \mbox{ and } & \\
& \LLL(\theta;(\x_1,y_1),\dots,(\x_n,y_n)) = \sum_{i=1}^n \ell(\theta;(\x_i,y_i)).&
\end{eqnarray*}

In Appendix~\ref{app:convex}, we review a few additional definitions related to convex functions and Gaussian concentration. For a set of vectors, we define its diameter as the maximum attained norm in the set.
\begin{definition} \label{defn:diameter}
(Diameter of Set) The diameter $\|\CCC\|$ of a closed set $\CCC \subseteq \R^d$, is defined as $\| \CCC \| = \sup_{\theta \in \CCC} \| \theta \|$.
\end{definition}

For improving the worst-case dependence on dimension $d$, we exploit the geometric properties of the input and constraint space. We use the well-studied quantity of {\em Gaussian width} that captures the $L_2$-geometric complexity of a set $S \subseteq \R^d$. 
\begin{definition} [Gaussian Width]
Given a closed set $S \subseteq \R^d$, its Gaussian width $w(S)$ is defined as: 
$$w(S)=\E_{\g \in \NNN(0,1)^d}\, [\sup_{\a \in S} \langle \a, \g \rangle].$$
\end{definition} 
In particular, $w(S)^2$ can be thought as the ``effective dimension'' of $S$. Many popular convex sets have low Gaussian width, e.g., the width of both the unit $L_1$-ball in $\R^d$ ($B_1^d$) and the standard $d$-dimensional probability simplex equals $\Theta(\sqrt{\log d})$, and the width of any ball $B_p^d$ for $1 \leq p \leq \infty$ is $\approx d^{1-1/p}$. For a set $\CCC$ contained in the $B_2^d$, $w(\CCC)$ is always $O(\sqrt{d})$. Another prominent set with low Gaussian width is that made up of sparse vectors. For example, the set of all $k$-sparse vectors (with at most $k$ non-zero entries) in $\R^d$ has Gaussian width $\Theta(\sqrt{k \log(d/k)})$.

\smallskip
\noindent\textbf{Differential Privacy on Streams.} We will consider differential privacy on data streams~\cite{DNPR10}. A stream is a sequence of points from some domain set $\ZZZ$. We say that two streams $\Gamma, \Gamma' \in \ZZZ^\ast$ of the same length are neighbors if there exists a datapoint $\z \in \Gamma$ and $\z' \in \ZZZ$ such that if we change $\z$ in $\Gamma$ to $\z'$ we get the stream $\Gamma'$. The result of an algorithm processing a stream is a sequence of outputs.

\begin{definition}[Event-level differential privacy~\cite{DMNS06,DNPR10}]\label{defn:dp}
Algorithm $\Alg$ is $(\epsilon,\delta)$-differentially private\footnote{In the practice of differential privacy, we generally think of $\eps$ as a small non-negligible constant, and $\delta$ as a parameter that is cryptographically small.} if for all neighboring streams $\Gamma, \Gamma'$ and for all sets of possible output sequences $\RRR \subseteq \R^\N$, we have
$$\Pr[\Alg(\Gamma) \in \RRR] \leq \exp(\epsilon) \cdot \Pr[\Alg(\Gamma') \in \RRR] + \delta,$$
where the probability is taken over the randomness of the algorithm.  When $\delta=0$, the algorithm $\Alg$ is $\eps$-differentially private.
\end{definition}

We provide additional background on differential privacy along with some techniques for achieving it in Appendix~\ref{app:dp}. 

\section{Private Incremental ERM: A Generic Mechanism}  \label{sec:generic}
We present a generic transformation for converting any private batch ERM algorithm into a private incremental ERM algorithm. We take this construction as a baseline for comparison for our private incremental regression algorithms. 

Mechanism~\PrivIncERM describes this simple transformation. At every timestep, Mechanism~\PrivIncERM outputs a $\theta^\priv_t$, a differentially private approximation of 
\begin{align*}
\hat{\theta}_t \in \mbox{argmin}_{\theta \in \CCC}\, \JJJ(\theta;\z_1,\dots,\z_t), \mbox{ where } \JJJ(\theta;\z_1,\dots,\z_t)  =  \sum_{i=1}^t \jmath(\theta;\z_i).
\end{align*}

The idea is to perform ``relevant'' computations only every $\tau$ timesteps, thereby ensuring that no $\z_i$ is used in more than $T/\tau$ invocations of the private batch ERM algorithm (for simplicity, assume that $T$ is a multiple of $\tau$). This idea is reminiscent of mini-batch processing ideas commonly used in big data processing~\cite{canny2013bidmach}. In Theorem~\ref{thm:thm1}, the parameter $\tau$ is set to balance the increase in excess risk due to lack of updates on the estimator and the increase in the excess risk due to the change in the privacy risk parameter $\eps$ (which arises from multiple interactions with the data).    

\begin{mechanism}[!t]
\DontPrintSemicolon
\caption{\PrivIncERM$(\eps,\delta)$}
\KwIn{A stream $\Gamma = \z_1,\dots,\z_T$, where each $\z_t$ is from the domain $\ZZZ \subseteq \R^d$, and $\tau \in \N $}
\KwOut{A differentially private estimate of $\hat{\theta}_t \in \mbox{argmin}_{\theta \in \CCC}\, \sum_{i=1}^t \jmath(\theta;\z_i)$ at every timestep $t \in [T]$}
Set $\eps' \leftarrow \frac{\eps}{\left (2\sqrt{\frac{2T}{\tau} \ln\left(\frac{2}{\delta}\right)} \right )}$ and  $\delta' \leftarrow \frac{\delta\tau}{2T}$\;
$\theta^\priv_0 \leftarrow \mathbf{0}$\;
\For{all $t \in [T]$}{
\If{$t$ is a multiple of $\tau$}{
$\theta^\priv_t \leftarrow$ Output of an $(\eps',\delta')$-differentially private algorithm minimizing $\JJJ(\theta;\z_1,\dots,\z_t)$ \label{step:dp}
}
\Else{
$\theta^\priv_t \leftarrow \theta^\priv_{t-1}$
}
Return $\theta^\priv_t$
}
\end{mechanism}

Mechanism~\PrivIncERM invokes a differentially private (batch) ERM algorithm for timesteps $t$ which are a multiple of $\tau$, and in all other timesteps it just outputs the result from the previous timestep. In Step~\ref{step:dp} of Mechanism~\PrivIncERM any differentially private batch ERM algorithm can be used, and this step dominates the time complexity of this mechanism. Here we present excess risk bounds obtained by invoking Mechanism~\PrivIncERM with the differentially private ERM algorithms of Bassily~\emph{et al.}\ \cite{bassily2014differentially} and Talwar~\emph{et al.}\ \cite{talwar2015nearly}. As mentioned earlier, the bounds of Bassily~\emph{et al.}\ \cite{bassily2014differentially} are tight in the worst-case. But as shown by Talwar~\emph{et al.}\ \cite{talwar2015nearly},  if the constraint space $\CCC$ has a small Gaussian width, then these bounds could be improved. The bounds of Talwar~\emph{et al.}\  depend on the curvature constant of $\jmath$ defined as:
\begin{flalign*}
C_\jmath = \sup_{\z \in \ZZZ} \sup_{\theta_a,\theta_b \in \CCC, l \in (0,1],\theta_c = \theta_a + l(\theta_b-\theta_a)} & \frac{2}{l^2} (\jmath(\theta_c;\z) - \jmath(\theta_a;\z) - \langle \theta_c-\theta_a,\nabla \jmath(\theta_a;\z) \rangle  ).
\end{flalign*}
For linear regression where $\jmath(\theta;\z) = \ell(\theta;(\x,y)) = (y - \langle \x,\theta\rangle)^2$, then with $\| \x \| \leq 1$ and $| y | \leq 1$, it follows that $C_\ell \leq \| \CCC \|^2$~\cite{clarkson2010coresets}. 

We now show that Mechanism~\PrivIncERM is event-level differentially private (Definition~\ref{defn:dp}), and analyze its utility under various invocations in Step~\ref{step:dp}. 
\begin{theorem} \label{thm:thm1}  Mechanism~\PrivIncERM is $(\eps,\delta)$-differentially private with respect to a single datapoint change in the stream~$\Gamma$. Also,
\begin{CompactEnumerate}
\item \label{res:incerm1} (Using Theorem 2.4, Bassily~\emph{et al.}\ \cite{bassily2014differentially}).  If the function $\jmath(\theta;\z) \,:\, \CCC \times \ZZZ \rightarrow \R$ is a positive-valued function that is convex with respect to $\theta$ over the domain $\CCC \subseteq \R^d$. Then for any $\beta > 0$, with probability at least $1-\beta$, for each $t \in [T]$, $\theta^\priv_t$ generated by Mechanism~\PrivIncERM with $\tau = \lceil \frac{(Td)^{1/3}}{\eps^{2/3}} \rceil$ satisfies:
\begin{align*}
\JJJ(\theta^\priv_t; \Gamma_t) - \min_{\theta \in \CCC}\, \JJJ(\theta; \Gamma_t)  = O \left ( \min \{ \frac{(Td)^{1/3} L \| \CCC \| \log^{5/2}(1/\delta) \, \polylog(T/\beta)}{\eps^{2/3}}, T L \| \CCC \| \} \right),
\end{align*}
where $L$ is the Lipschitz parameter of the function $\jmath$.
\item \label{res:incerm2} (Using Theorem 2.4, Bassily~\emph{et al.}\ \cite{bassily2014differentially}). If the function $\jmath(\theta;\z) \,:\, \CCC \times \ZZZ \rightarrow \R$ is a positive-valued function which is $\nu$-strongly convex with respect to  $\theta$ over the domain $\CCC \subseteq \R^d$. Then for any $\beta > 0$, with probability at least $1-\beta$, for each $t \in [T]$, $\theta^\priv_t$ generated by Mechanism~\PrivIncERM with $\tau = \lceil \frac{\sqrt{d L}}{(\nu^{1/2} \eps \| \CCC \|^{1/2})} \rceil $ satisfies:
\begin{align*}
\JJJ(\theta^\priv_t; \Gamma_t) - \min_{\theta \in \CCC}\, \JJJ(\theta; \Gamma_t) = O \left ( \min \{\frac{\sqrt{d} L^{3/2} \| \CCC \|^{1/2} \log^4(1/\delta)\polylog(T/\beta)}{\nu^{1/2} \eps}, TL \| \CCC \| \} \right ),
\end{align*}
where $L$ is the Lipschitz parameter of the function $\jmath$.
\item \label{res:incerm3} (Using Theorem 2.6 of Talwar~\emph{et al.}\ \cite{talwar2015nearly}). If the function $\jmath(\theta;\z) \,:\, \CCC \times \ZZZ \rightarrow \R$ is a positive-valued function that is convex with respect to $\theta$ over the domain $\CCC \subseteq \R^d$. Then for any $\beta > 0$, with probability at least $1-\beta$, for each $t \in [T]$, $\theta^\priv_t$ generated by Mechanism~\PrivIncERM with $\tau = \lceil \frac{\sqrt{Tw(\CCC)}C_\jmath^{1/4}}{((L\|\CCC\|)^{1/4}\eps^{1/2})} \rceil$ satisfies:
\begin{align*}
\JJJ(\theta^\priv_t; \Gamma_t) - \min_{\theta \in \CCC}\, \JJJ(\theta; \Gamma_t) = O \left (\min \{ \frac{\sqrt{Tw(\CCC)}C_\jmath^{1/4} (L \| \CCC \|)^{3/4} \log^{7/3}(1/\delta) \, \polylog(T/\beta)}{\eps^{1/2}}, T L \| \CCC \| \} \right ),
\end{align*}
where $L$ is the Lipschitz parameter and $C_\jmath$ is the curvature constant of the function $\jmath$.
\end{CompactEnumerate}
\end{theorem}
\begin{proof}
\noindent\textbf{Privacy Analysis.} Each $\z_i$ is accessed at most $T/\tau$ times by the algorithm invoked in Step~\ref{step:dp}. Let $l=T/\tau$. By using the composition Theorem~\ref{thm:composition2} with $\delta^\ast = \delta/2$ it follows that the entire algorithm is $(\eps'\sqrt{2l\ln(2/\delta)}+2l\eps'^2,l\cdot(\delta/2l)+\delta/2)$-differentially private. We set $\eps'$ as $\eps/2 = \eps'\sqrt{2l\ln(2/\delta)}$. Note that with this setting of $\eps'$, $2l\eps'^2 \leq \eps/2$, therefore, $\eps'\sqrt{2l\ln(2/\delta)}+2l\eps'^2 \leq \eps$. Hence, Mechanism~\PrivIncERM is $(\eps,\delta)$-differentially private.

\noindent\textbf{Utility Analysis.} If $T < \tau$, algorithm does not access the data, and in that case, the excess risk can be bounded by $ TL \| \CCC \|$. Now assume $T\geq \tau$. Note that the algorithm performs no computation when $t$ is not a multiple of $\tau$. Let $\Gamma_t$ denote the prefix of stream $\Gamma$ till time $t$. Let $t$ lie in the interval of $[j\tau,(j+1)\tau]$ for some $j \in \N$. The total loss accumulated by the algorithm at time $t$ can be split as:
\begin{align*}
\sum_{i=1}^t \jmath(\theta^\priv_t;\z_i) & = \sum_{i=1}^{j\tau} \jmath(\theta^\priv_{j\tau};\z_i) + \sum_{i=j\tau+1}^t \jmath(\theta^\priv_{j\tau};\z_i) \leq \sum_{i=1}^{j\tau}  \jmath(\theta^\priv_{j\tau};\z_i) + \tau L \| \CCC\|,
\end{align*}
as $\theta^\priv_t = \theta^\priv_{j \tau}$. 

Let $\hat{\theta}_t \in \mbox{argmin}_{\theta \in \CCC}\, \sum_{i=1}^t \jmath(\theta;\z_i)$. As $\jmath$ is positive-valued,
$$\sum_{i=1}^t \jmath(\hat{\theta}_t;\z_i) \geq \sum_{i=1}^{j\tau} \jmath(\hat{\theta}_t;\z_i) \geq \sum_{i=1}^{j\tau} \jmath(\hat{\theta}_{j \tau};\z_i).$$
Hence, we get,
\begin{align*}
\sum_{i=1}^t \jmath(\theta^\priv_t;\z_i) - \sum_{i=1}^t \jmath(\hat{\theta}_t;\z_i)  \leq \sum_{i=1}^{j\tau} \jmath(\theta^\priv_{j\tau};\z_i) - \sum_{i=1}^{j\tau}\jmath(\hat{\theta}_{j \tau};\z_i) + \tau L \| \CCC\|.
\end{align*}
Using the results from Bassily~\emph{et al.}\ \cite{bassily2014differentially} or Talwar~\emph{et al.}\ \cite{talwar2015nearly} to bound $\sum_{i=1}^{j \tau} \jmath(\theta^\priv_{j\tau};\z_i)  - \sum_{i=1}^{j \tau} \jmath(\hat{\theta}_{j \tau};\z_i)$, setting $\tau$ to balance the various opposing terms, and a final union bound provide the claimed bounds.
\end{proof}

\section{Private Incremental Linear Regression using Tree Mechanism} \label{sec:smalld}
We now focus on the problem of private incremental linear regression. Our first approach for this problem is based on a private incremental computation of the gradient. The algorithm is particularly effective in the regime of large $T$ and small $d$. A central idea of our approach is the construction of a {\em private gradient function} defined as follows.
\begin{definition} \label{defn:privgradfunc}
Let $\CCC \subseteq \R^d$. Algorithm $\Alg$ computes an $(\alpha,\beta)$-accurate gradient of the loss function $\JJJ(\theta;\z_1,\dots,\z_t)$ with respect to $\theta \in \CCC$, if given $\z_1,\dots,\z_t \in \ZZZ$ it outputs a function $g_t: \CCC \rightarrow \R^d$ such that:
\begin{list}{{\bf (\roman{enumi})}}{\usecounter{enumi}
\setlength{\leftmargin}{\parindent}
\setlength{\listparindent}{\parindent}
\setlength{\parsep}{0pt}}
\item \label{res:privacy} {\bf Privacy}: $\Alg$ is $(\eps,\delta)$-differentially private (as in Definition~\ref{defn:dp}), i.e., for all neighboring streams $\Gamma, \Gamma' \in \ZZZ^\ast$ and subsets $\RRR \subseteq \CCC \rightarrow \R^d$,
$$\Pr[\Alg(\Gamma) \in \RRR] \leq \exp(\eps) \cdot \Pr[\Alg(\Gamma') \in \RRR]+\delta.$$
\item \label{res:utility} {\bf Utility}: The function $g_t$ is an $(\alpha,\beta)$-approximation to the true gradient, in that, 
$$\Pr_{\Alg}\left[\max_{\z_1,\dots,\z_t \in\ZZZ, \theta\in\CCC}\| g_t(\theta) - \nabla \JJJ(\theta;\z_1,\dots,\z_t)  \| \geq \alpha\right]\leq\beta.$$
\end{list}
\end{definition}
Note that the output of $\Alg$ in the above definition is a function $g_t$. The first requirement on $\Alg$ specifies that it satisfies the differential privacy condition (Definition~\ref{defn:dp}). The second requirement on $\Alg$ specifies that for any $\theta \in \CCC$, it gives a ``sufficiently'' accurate estimate of the true gradient $\nabla \JJJ(\theta;\z_1,\dots,\z_t)$.

Let $\Gamma = (\x_1,y_1),\dots,(\x_T,y_T)$ represent the stream of covariate-response pairs, we use $\Gamma_t$ to denote $(\x_1,y_1),\dots,(\x_t,y_t)$. Consider the gradient of the loss function $\LLL(\theta; \Gamma_t)$ where $X_t \in \R^{t \times d}$ is a matrix with rows $\x_1^\top,\dots,\x_t^\top$ and $\y_t=(y_1,\dots,y_t)$:
\begin{align} \label{eqn:gradvec}
\nabla \LLL(\theta; \Gamma_t) =  2 (X_t^\top X_t \theta - X_t^\top \y_t) = 2 \left (\sum_{i=1}^t \x_i \x_i^\top \theta  -  \sum_{i=1}^t \x_i y_i\right ).
\end{align}
A simple observation from the gradient form of~\eqref{eqn:gradvec} is that if we can maintain the streaming sum of $\x_1 y_1,\dots,\x_T y_T$ and the streaming sum of $\x_1 \x_1^\top,\dots,\x_T \x_T^\top$, then we can maintain the necessary ingredients needed to compute $\nabla \LLL(\theta; \Gamma_t)$ for any $t \in [T]$. We use this observation to construct a private gradient function $g_t : \CCC \rightarrow \R^d$ at every timestep $t$. The idea is to privately maintain $\sum_{i=1}^T \x_i \x_i^\top$ and $\sum_{i=1}^T \x_i y_i$ over the steam using the \emph{Tree Mechanism} of~\cite{DNPR10,CSS11}. We present the entire construction of the \emph{Tree Mechanism} in Appendix~\ref{app:treemech}. The rough idea behind this mechanism is to build a binary tree where the leaves are the actual inputs from the stream, and the internal nodes store the partial sums of all the leaves in its sub-tree. For a stream of vectors, $\upsilon_1,\dots,\upsilon_T$ in the unit ball, the \emph{Tree Mechanism} allows private estimation of $\sum_{i=1}^t \upsilon_i$, for every $t \in [T]$, with error roughly $\sqrt{d} \log^{2} T$ (under $(\eps,\delta)$-differential privacy, ignoring other parameters). 

Somewhat similar to our approach, Smith and Thakurta also use the \emph{Tree Mechanism} to maintain the sum of gradients in their private online learning algorithm~\cite{thakurta2013nearly}, however, unlike our approach, they do not construct a private gradient function.

\begin{algorithm}[!t]
\DontPrintSemicolon
\caption{\PrivIncReg$(\eps,\delta)$}
\KwIn{A stream $\Gamma = (\x_1,y_1),\dots,(\x_T,y_T)$, where each $(\x_t,y_t)$ in $\Gamma$ is from the domain $\XXX \times \YYY$ where $\XXX \subset \R^d$ with $\| \XXX \| \leq 1$ and  $\YYY \subset \R$ with $\| \YYY \| \leq 1$}
\KwOut{A differentially private estimate of $\hat{\theta}_t \in \mbox{argmin}_{\theta \in \CCC}\, \sum_{i=1}^t (y_i - \langle \x_i,\theta \rangle)^2$ at every timestep $t \in [T]$}
Set $\eps' \leftarrow \frac{\eps}{2}$, \,\, $\delta' \leftarrow \frac{\delta}{2}$, \,\, $\kappa \leftarrow \frac{\log^{3/2}(T)  \sqrt{\log \left ( \frac{1}{\delta'} \right ) }}{\eps'}$,\,\, $\alpha' \leftarrow O(\kappa \| \CCC \| \sqrt{d})$, \, and $r \leftarrow \Theta \left ( \left (1+ \frac{T\|\CCC\|}{\alpha'} \right )^2 \right )$\;
\For{all $t \in [T]$}{
$\q_t \leftarrow $ output of \TreeMech$(\eps',\delta',2)$ at $t$ when invoked on the stream $\x_1 y_1,\dots,\x_T y_T$ \label{Step3}\;
$Q_t \leftarrow $ output of \TreeMech$(\eps',\delta',2)$ at $t$ when invoked on the stream $\x_1 \x_1^\top,\dots,\x_T \x_T^\top$ which can be viewed as $d^2$-dimensional vectors (the outputs are converted back to form $d \times d$ matrices)\label{Step4}\;
Define a private gradient function, $g_t : \CCC \rightarrow \R^d$ as:
$$g_t(\theta) =  2(Q_t \theta - \q_t)$$ \label{Step5}
\vspace*{-3ex}
\;
$\theta^\priv_t \leftarrow \NoisyProjGrad(\CCC,g_t,r)$ (described in Appendix~\ref{app:noisyproj}) \label{Step7}\;
Return $\theta^\priv_t$
}
\end{algorithm}

Once the private gradient function $g_t$ is released, it could be used to obtain an approximation to the estimator $\hat{\theta}$ by using any traditional gradient-based optimization technique. In this paper, we use a variant of the classical {\em projected gradient descent} approach, described in Algorithm \NoisyProjGrad (Appendix~\ref{app:noisyproj}). Algorithm \NoisyProjGrad is a iterative algorithm that takes in the constraint set $\CCC$, private gradient function $g_t$, and a parameter $r$ denoting the number of iterations, and returns back ($\theta_t^\priv$) a private estimate of regression parameter.

An important point to note is that evaluating the gradient function at different $\theta$'s, as needed for any gradient descent technique, does not affect the privacy parameters $\eps$ and $\delta$. We analyze the convergence of Algorithm \NoisyProjGrad in Appendix~\ref{app:noisyproj}. This convergence result (Corollary~\ref{cor:noisyprojgrad}) is used to set the parameter $r$ in Algorithm~\PrivIncReg. 

Algorithm~\PrivIncReg is $(\eps,\delta)$-differentially private with respect to a single datapoint change in the stream $\Gamma$. This follows as the $L_2$-sensitivity 
of the stream in both the invocations (Steps~\ref{Step3} and~\ref{Step4}) of Algorithm~\TreeMech is less than $2$ (because of the normalization on the $\x_i$'s and $y_i$'s). Using the standard composition properties (Theorem~\ref{thm:composition1}) of differential privacy gives that the Algorithm~\PrivIncReg is $(\eps,\delta)$-differentially private. The algorithm only requires $O(d^2 \log T)$ space (see Appendix~\ref{app:treemech}), therefore having only a logarithmic dependence on the length of the stream. The running time of the algorithm is dominated by Steps~\ref{Step4} and~\ref{Step7}; at every timestep $t$, Step~\ref{Step4} has time complexity of $O(d^2 \log T)$, whereas the time complexity of Step~\ref{Step7} is $r$ times the time complexity of projecting a datapoint onto $\CCC$ (the $P_\CCC$ operation defined in Appendix~\ref{app:noisyproj}). We now analyze the utility of this algorithm, using the error bound on \emph{Tree Mechanism} from Proposition~\ref{prop:counter}.

\begin{lemma} \label{lem:binmecherror}
For any $\beta > 0$, $\theta \in \CCC$, and $t \in [T]$, with probability at least $1-\beta$, the function $g_t$ defined in Algorithm~\PrivIncReg satisfies:
$$ \left \| g_t(\theta) - \nabla \LLL(\theta;\Gamma_t) \right \| = O \left ( \kappa \| \CCC \| (\sqrt{d}+\sqrt{\log(1/\beta)}) \right ).$$
\end{lemma}
\begin{proof}
Applying Proposition~\ref{prop:counter}, we know with probability at least $1-\beta$,
\begin{eqnarray*}
& \left \| \q_t - \sum_{i=1}^t \x_i y_i  \right \|   =  O \left ( \kappa  (\sqrt{d}+\sqrt{\log(1/\beta)}) \right ) \mbox{ and } \left \| Q_t - \sum_{i=1}^t \x_i \x_i^\top \right \|   = O \left ( \kappa (\sqrt{d}+\sqrt{\log(1/\beta)})\right ).&
\end{eqnarray*}
Therefore, we get with probability at least $1-\beta$,
\begin{align*}
\| g_t(\theta) - \nabla \LLL(\theta;\Gamma_t) \|  =  \left \| 2 \left ( Q_t - \sum_{i=1}^t \x_i \x_i^\top \right ) \theta  - 2 \left (\q_t - \sum_{i=1}^t \x_i y_i  \right )\right \| =  O \left ( \kappa \| \CCC \| (\sqrt{d}+\sqrt{\log(1/\beta)}) \right ),
\end{align*}
where we used the fact that $\| \theta \| \leq \| \CCC \|$ for any $\theta \in \CCC$.
\end{proof}

\begin{theorem} \label{thm:thm2}
Algorithm~\PrivIncReg is $(\eps,\delta)$-differentially private with respect to a single datapoint change in the stream~$\Gamma$. For any $\beta > 0$, with probability at least $1-\beta$, for each $t \in [T]$, $\theta^\priv_t$ generated by Algorithm~\PrivIncReg satisfies:
\begin{align*}
\LLL(\theta^\priv_t; \Gamma_t)  - \min_{\theta \in \CCC}\, \LLL(\theta; \Gamma_t) =  O \left ( \frac{\log^{3/2} T \sqrt{\log(1/\delta)} \| \CCC \|^2 (\sqrt{d}+\sqrt{\log(T/\beta)})}{\eps}\right ).
\end{align*}
\end{theorem} 
\begin{proof}
The $(\eps,\delta)$-differential privacy follows from the above discussed global sensitivity bound.

Fix any $t \in [T]$. Let $\hat{\theta}_t \in \mbox{argmin}_{\theta \in \CCC}\, \sum_{i=1}^t (y_i - \langle \x_i,\theta \rangle)^2$. Combining Lemma~\ref{lem:binmecherror} and Corollary~\ref{cor:noisyprojgrad}, with probability at least $1-r\beta'$
$$\LLL(\theta^\priv_t; \Gamma_t) - \LLL(\hat{\theta}_t; \Gamma_t) = O\left ( \kappa \| \CCC \|^2 (\sqrt{d}+\sqrt{\log(1/\beta')})\right).$$
Replacing $\beta'$ by $\beta/r$, substituting for $\kappa$, and taking a union bound over all $t \in [T]$ gives the claimed result.
\end{proof}

\begin{remark}
For linear regression instances, which typically do not satisfy the strong convexity property, the $\approx \min\{(Td)^{1/3},T\}$ risk bound obtained from Mechanism~\PrivIncERM is substantially worse than the $\approx \min\{\sqrt{d},T\}$ risk bound provided by Algorithm \PrivIncReg.  

The dependence on dimension $d$ is tight in the worst-case; due to $\approx \sqrt{d}$ excess risk lower bounds established by Bassily~\emph{et al.}\ \cite{bassily2014differentially}. 
\end{remark}

\begin{remark}
The techniques developed in this section (based on Tree Mechanism) can be applied to any convex ERM problem whose gradient has a {\em linear} form, in which case we obtain an excess risk bound as in Theorem~\ref{thm:thm2}.  It is an interesting open question to obtain similar bounds for general convex ERM problems in an incremental setting. 
\end{remark}

\section{Private Incremental Linear Regression: Going Beyond Worst-Case} \label{sec:width}
The noise added in Theorem~\ref{thm:thm2} for privacy grows proportionately as $\sqrt{d}$. While this seems unavoidable in the worst-case, we ask whether it is possible to go beyond this worst-case bound under some realistic assumptions. An intuitive idea to overcome this curse of dimensionality is to project (compress) the data to a lower dimension, before the addition of noise. In this section, we use this and other geometric ideas to cope with the high-dependency on dimensionality in Theorem~\ref{thm:thm2}. The resulting bound will depend on the Gaussian width of the input and constraint space which for many interesting problem instances will be smaller than $\sqrt{d}$. We mention some specific instantiations and extensions of our result in Section~\ref{sec:inst} including a scenario when not all inputs are drawn from a domain with a ``small'' Gaussian width. 

Our general approach in this section will be based on this simple principle: reduce the dimensionality of the problem, solve it privately in the lower dimensional space, and then ``lift'' the solution back to the original space. Our lifting procedure is similar to that used by Kasiviswanathan and Jin~\cite{kasiviswanathanicml} in their recent work on bounding excess risk for private ERM in a batch setting. 

Fix $t$ and consider solving the following projected least-squares problem, defined as:\!\footnote{The scaling factor of $\frac{\| \x_i \|}{\| \Phi \x_i \|}$ is for simplicity of analysis only. One could omit it and still obtain the same results using a slightly different analysis. Also without loss of generality, we assume $\x_i \neq 0$ for all $i$.}
\begin{equation}\label{eqn:JLtransform}
\LLL_\proj(\theta;\Gamma_t;\Phi) = \sum_{i=1}^t \left (y_i - \frac{\| \x_i \|}{\| \Phi \x_i \|} \langle \Phi \x_i, \Phi \theta \rangle \right )^2,
\end{equation}
where $\Phi \in \R^{m \times d}$ is a random projection matrix (to be defined later). The loss function, $\LLL_\proj$ is also referred to as {\em compressed least-squares} in the literature~\cite{maillard2009compressed,fard2012compressed,kaban2014new}.

As a first step, we investigate the relationship between $\LLL(\theta;\Gamma_t)$ and $\LLL_\proj(\theta;\Gamma_t;\Phi)$. A fundamental tool in dimensionality reduction, the Johnson-Lindenstrauss (JL) lemma, states that for any set $S \subseteq \R^d$, given  $\gamma > 0$ and $m = \Omega(\log |S|/\gamma^2)$, there exists a map that embeds the set into $\R^m$, distorting all pairwise distances within at most $1 \pm \gamma$ factor. A simple consequence of the JL Lemma is that, for any set of $n$ vectors $\x_1,\dots,\x_n \in \R^d$, $\theta \in \R^d$, $\gamma > 0$, and $\beta > 0$, if $\Phi$ is an $m \times d$ matrix with entries drawn i.i.d.\ from $\NNN(0,1/m)$ with $m = \Theta(\log(n/\beta)/\gamma^2)$, we get, 
\begin{equation} \label{eqn:JLprob}
\Pr\left[| \langle \Phi \x_i, \Phi \theta \rangle - \langle \x_i, \theta \rangle | \geq \gamma \| \x_i \| \| \theta \|~\mbox{\rm for all}~i\in[n]\right] \leq \beta.
\end{equation}
Using standard inequalities such as~\eqref{eqn:JLprob}, establishing the relationship between $\LLL$ and $\LLL_\proj$ is relatively straightforward in a batch setting. However, the incremental setting raises the challenge of dealing with adaptive inputs in JL transformation. The issue being that JL transformation is inherently non-adaptive, in that ``success'' of the JL transformation (properties such as~\eqref{eqn:JLprob}) depends on the fact that the inputs on which it gets
applied are chosen before (independent of) fixing the transformation $\Phi$.


To deal with this kind of adaptive generation of inputs, we use a generalization of JL lemma which yields similar guarantees but with a much smaller (than $\log |S|$) requirement on $m$, if the set $S$ has certain geometric characteristics. For ease of exposition, in the rest of this section, we are going to assume that $\Phi$ is a matrix in $\R^{m \times d}$ with i.i.d.\ entries from $\NNN(0,1/m)$.\footnote{One could also use other (better) constructions of $\Phi$, such as those that create sparse $\Phi$ matrix, using recent results by Bourgain~\emph{et al.}\ \cite{bourgain2015toward} extending Theorem~\ref{thm:gordon} to other distributions.} Gordon~\cite{gordon1988milman} showed that using a Gaussian random matrix, one can embed the set $S$ into a lower-dimensional space $\R^m$, where $m$ roughly scales as the square of Gaussian width of $S$. This result has found several interesting applications in high-dimensional convex geometry, statistics, and compressed sensing~\cite{pfandersampling}.

\begin{theorem}[Gordon~\cite{gordon1988milman}]\label{thm:gordon}
Let $\tilde{\Phi}$ be an $m \times d$ random matrix, whose rows $\phi_1^\top,\dots,\phi_m^\top$ are i.i.d.\ Gaussian random vectors in $\R^d$ chosen according to the standard normal distribution $\NNN(\mathbf{0},\mathbb{I}_d)$. Let $\Phi = \tilde{\Phi}/\sqrt{m}$. Let $S$ be a set of points in $\R^d$. There is a constant $C > 0$ such that for any $0 < \gamma, \beta < 1$,
$$\Pr \left [\sup_{\a \in S} \left | \| \Phi \a \|^2 - \| \a \|^2 \right | \geq  \gamma \| \a \|^2 \right ] \leq \beta,$$ 
provided that $m \geq \frac{C}{\gamma^2} \max \left \{ w(S)^2, \log \left (\frac{1}{\beta} \right \} \right )$.
\end{theorem}
Note that the $w(S)$ is defined for all sets, not just convex sets, a fact that we use below as the input domain $\XXX$ may not be convex.  As a simple corollary to the above theorem it also follows that,
\begin{corollary} \label{cor:gordon}
Under the setting of Theorem~\ref{thm:gordon}, there exists a constant $C' > 0$ such that for any $0 < \gamma, \beta < 1$,
$$\Pr \left [\sup_{\a,\b \in S} \left | \langle \Phi \a, \Phi \b \rangle - \langle \a,\b \rangle  \right | \geq \gamma \| \a \| \| \b \| \right ] \leq \beta,$$ 
provided that $m \geq \frac{C'}{\gamma^2} \max \left \{ w(S)^2, \log \left (\frac{1}{\beta} \right \} \right )$.
\end{corollary}
Applying the above corollary to the set of vectors in $\XXX \cup \CCC$, and by noting that $w(\XXX \cup \CCC) \leq w(\XXX) + w(\CCC)$, gives that if $m = \Theta((1/\gamma^2)\max\{(w(\XXX)+w(\CCC))^2,\log(1/\beta)\})$, then with probability at least $1-\beta$,
\begin{align} \label{eqn:xthetprod}
\Pr \left [\sup_{\x \in \XXX,\theta \in \CCC} \left | \langle \Phi \x, \Phi \theta \rangle - \langle \x,\theta \rangle  \right | \geq \gamma \| \x \| \| \theta \| \right ] \leq \beta.
\end{align}

\subsection{Algorithm for the Streaming Setting} \label{sec:algstream}
We now present a mechanism (Algorithm~\ProjPrivIncReg) for private incremental linear regression based on minimizing the projected least-squares objective~\eqref{eqn:JLtransform}, under differential privacy. The idea is to again construct a private gradient function $g_t$, but of function $\LLL_\proj$ (instead of $\LLL$ as done in Algorithm~\PrivIncReg). 

Let $X_t$ be a matrix with rows $\x_1^\top,\dots,\x_t^\top$, and $\tilde{X}_t \in \R^{n \times d}$ be a matrix with rows $\tilde{\x}_1^\top,\dots,\tilde{\x}_t^\top$. As before, let $\y_t$ be the vector $(y_1,\dots,y_t)$. Under these notation, $\LLL_\proj(\theta;\Gamma_t;\Phi)$ from~\eqref{eqn:JLtransform} can be re-expressed as:
$$\LLL_\proj(\theta;\Gamma_t;\Phi) = \| \y_t - \tilde{X}_t \Phi^\top\Phi \theta\|^2.$$
The gradient of $\LLL_\proj$ with respect to $\Phi \theta$ equals:
\begin{align*}
& \nabla_{(\Phi\theta)} \LLL_\proj(\theta;\Gamma_t;\Phi)  = \frac{\partial \; \| \y_t - (\tilde{X}_t \Phi^\top) (\Phi \theta)\|^2}{\partial \; (\Phi \theta)} =  2((\tilde{X}_t\Phi^\top)^\top (\tilde{X}_t\Phi^\top)) (\Phi \theta) - 2(\tilde{X}_t\Phi^\top)^\top \y_t.
\end{align*}
Note that $\nabla_{(\Phi\theta)} \LLL_\proj \in \R^m$.  

Let $\Phi\CCC = \{ \vartheta \in \Phi\theta \,: \, \theta \in \CCC\}$. Note for a convex $\CCC$, $\Phi\CCC \subset \R^m$ is also convex. In Algorithm~\ProjPrivIncReg, $\vartheta^\priv_t$ is a private estimate of $\hat{\vartheta}_t$, where
$$\hat{\vartheta}_t \in  \mbox{argmin}_{\vartheta \in \Phi\CCC}\,  \sum_{i=1}^t \left (y_i - \frac{\| \x_i \|}{\| \Phi \x_i \|} \langle \Phi \x_i,\vartheta \rangle \right )^2.$$

\begin{algorithm}[!t]
\DontPrintSemicolon
\caption{\ProjPrivIncReg$(\eps,\delta)$}
\KwIn{A stream $\Gamma = (\x_1,y_1),\dots,(\x_T,y_T)$, where each $(\x_t,y_t)$ in $\Gamma$ is from the domain $\XXX \times \YYY$ where $\XXX \subset \R^d$ with $\| \XXX \| \leq 1$ and  $\YYY \subset \R$ with $\| \YYY \| \leq 1$}
\KwOut{$\theta^\priv_t$ a differentially private estimate of $\hat{\theta}_t \in \mbox{argmin}_{\theta \in \CCC}\, \sum_{i=1}^t (y_i - \langle \x_i,\theta \rangle)^2$ at every timestep $t \in [T]$}
Set $\eps' \leftarrow \frac{\eps}{2}$, \,\, $\delta' \leftarrow \frac{\delta}{2}$, \,\, $\kappa \leftarrow \frac{\log^{3/2}(T)  \sqrt{\log \left ( \frac{1}{\delta'} \right ) }}{\eps'}$,\,\, $\alpha' \leftarrow O(\kappa \| \CCC \| \sqrt{m})$, \,\, $r \leftarrow \Theta \left ( \left (1+ \frac{T\|\CCC\|}{\alpha'} \right )^2 \right )$,\,\, $\gamma \leftarrow \frac{(w(\XXX)+w(\CCC))^{1/3}}{T^{1/3}}$, and $m \leftarrow \Theta \left (\frac{1}{\gamma^2}\max\{(w(\XXX)+w(\CCC))^2,\log(\frac{T}{\beta})\} \right )$\;
Let $\Phi \leftarrow m \times d$ random matrix with entries drawn i.i.d.\ from $\NNN(0,1/m)$\;
\For{all $t \in [T]$}{
Let $\tilde{\x}_t \leftarrow \frac{\| \x_t \|}{\| \Phi \x_t \|} \x_t$\;
$\q_t \leftarrow $ output of \TreeMech$(\eps',\delta',2)$ at $t$ when invoked on the stream $\Phi\tilde{\x}_1 y_1,\dots,\Phi\tilde{\x}_T y_T$ \label{step:first}\;
$Q_t \leftarrow $ output of \TreeMech$(\eps',\delta',2)$ at $t$ when invoked on the stream $(\Phi\tilde{\x}_1) (\Phi\tilde{\x}_1)^\top,\dots,(\Phi\tilde{\x}_T) (\Phi\tilde{\x}_T)^\top$ which can be viewed as $m^2$-dimensional vectors (the outputs are converted back to form $m \times m$ matrices) \label{step:second}\;
Define a private gradient function, $g_t : \Phi\CCC \rightarrow \R^m$ as:
\vspace*{-1ex}
$$g_t(\vartheta) =  2(Q_t \vartheta - \q_t)$$
\vspace*{-2ex}
\;
$\vartheta^\priv_t \leftarrow \NoisyProjGrad(\Phi\CCC,g_t,r)$ (described in Appendix~\ref{app:noisyproj})\;
$\theta^\priv_t \leftarrow \mbox{argmin}_{\theta \in \R^d}\, \| \theta \|_\CCC$ subject to $\Phi\theta = \vartheta^\priv_t$  (can be solved using any convex optimization technique)\label{step:lp} \;
Return $\theta^\priv_t$
}
\end{algorithm}

Algorithm~\ProjPrivIncReg only requires $O(m^2\log T+ \log d)$ space, therefore is slightly more memory efficient than Algorithm~\PrivIncReg (as $m \leq d$). The time complexity can be analyzed as for Algorithm~\PrivIncReg.  

Algorithm~\ProjPrivIncReg is $(\eps,\delta)$-differentially private with respect to a single datapoint change in the stream $\Gamma$. The $L_2$-sensitivity for both invocations of  Algorithm~\TreeMech is $2$. In Step~\ref{step:second}, this holds because
\begin{align*}
\max_{\x_a,\x_b \in \XXX}\, \| (\Phi \tilde{\x}_a)(\Phi \tilde{\x}_a)^\top - (\Phi \tilde{\x}_b)(\Phi \tilde{\x}_b)^\top \|_F & \leq \|(\Phi \tilde{\x}_a)(\Phi \tilde{\x}_a)^\top \|_F  + \| (\Phi \tilde{\x}_b)(\Phi \tilde{\x}_b)^\top \|_F \\
& = \|\Phi \tilde{\x}_a \|^2 + \|\Phi \tilde{\x}_b \|^2 = \| \x_a \|^2 + \| \x_b \|^2 = 2,
\end{align*}
the second to last equality follows as, for every $\x \in \XXX$, $\| \Phi \tilde{\x} \| =  \| \x \|$ (by construction). Since $\Phi \x_i$'s are in the projected subspace ($\R^m$), the noise needed for differential privacy (in Steps~\ref{step:first} and~\ref{step:second}) of Algorithm~\ProjPrivIncReg roughly scales as $\sqrt{m}$. 


In Step~\ref{step:lp} of Algorithm~\ProjPrivIncReg, we lift $\vartheta^\priv_t$ into the original $d$-dimensional constraint space $\CCC$. Since $\vartheta^\priv_t \in \Phi\CCC$, we know that there exists a $\theta^\true_t \in \CCC$, such that  $\Phi\theta^\true_t = \vartheta^\priv_t$. Then the goal is to estimate $\theta^\true_t$ from $\Phi\theta^\true_t$. Again geometry of $\CCC$ (Gaussian width) plays an important role, as it controls the diameter of high-dimensional random sections of $\CCC$ (referred to as $M^\star$ bound~\cite{ledoux2013probability,vershynin2014estimation}). We refer the reader to the excellent tutorial by Vershynin~\cite{vershynin2014estimation} for more details.

We define Minkowski functional, as commonly used in geometric functional analysis and convex analysis.
\begin{definition} [Minkowski functional]
For any vector $\theta \in \R^d$, the Minkowski functional of $\CCC$ is the non-negative number $\| \theta \|_\CCC$ defined by the rule: $\| \theta \|_\CCC = \inf \{\rho \in \R \,:\, \theta \in \rho \CCC \}$.
\end{definition}
For the typical situation in ERM problems, where $\CCC$ is a symmetric convex body, $\| \cdot \|_{\CCC}$ defines a norm. The optimization problem solved in Step~\ref{step:lp} of Algorithm~\ProjPrivIncReg is convex if $\CCC$ is convex, and hence can be efficiently solved. The existence of $\theta^\priv_t$ follows from the following theorem.

\begin{theorem}~\cite{vershynin2014estimation} \label{thm:exist}
Let $\Phi$ be an $m \times d$ matrix, whose rows $\phi_1^\top,\dots,\phi_m^\top$ are i.i.d.\ Gaussian random vectors in $\R^d$ chosen according to the standard normal distribution $\NNN(0,\mathbb{I}_d)$. Let $\CCC$ be a convex set. Given $\v = \Phi \u$ and $\Phi$, let $\hat{\u}$ be the solution to the following convex program: $\min_{\u'\in \R^d}\, \| \u' \|_{\CCC}$ subject to $\Phi \u' = \v$. Then for any $\beta > 0$, with probability at least $1-\beta$, 
$$\sup_{\u: \v = \Phi \u} \| \u - \hat{\u} \| = O\left ( \frac{w(\CCC)}{\sqrt{m}} + \frac{\| \CCC \| \sqrt{\log(1/\beta)}}{\sqrt{m}} \right).$$ 
\end{theorem}

The next thing to be verified is that $\theta^\priv_t$ generated by Algorithm~\ProjPrivIncReg is in $\CCC$. This is simple as by definition of Minkowski functional, as any closed set $\CCC = \{ \theta \in \R^d \,: \, \| \theta \|_{\CCC} \leq 1\}$. Hence, $\|\theta^\true_t\|_{\CCC} \leq 1$. By choice of $\theta^\priv_t$ in Step~\ref{step:lp}, ensures that $\| \theta^\priv_t\|_{\CCC} \leq \|\theta^\true_t\|_{\CCC} \leq 1$, which guarantees that $\theta^\priv_t \in \CCC$. Finally, note that the lifting is a post-processing operation on a differentially private output $\vartheta^\priv_t$, and hence does not affect the differential privacy guarantee.


\smallskip
\noindent\textbf{Utility Analysis of Algorithm~\ProjPrivIncReg.} In Lemma~\ref{lem:JLzeroth}, by using the fact that the Gaussian noise for privacy (in the \emph{Tree Mechanism}) is added on a lower dimensional ($m$) instance, we show that the difference between $\LLL_\proj(\theta^\priv_t;\Gamma_t;\Phi)$ and $\LLL_\proj(\hat{\theta}_t;\Gamma_t;\Phi)$ (the minimum empirical risk) roughly scales as $\sqrt{m}$, for sufficiently large $m$. The Lipschitz constant of the function $\LLL_\proj(\theta;\Gamma_t;\Phi)$ is $O(\| \Phi \CCC\|)$, which by Theorem~\ref{thm:gordon}, with probability at least $1-\beta$ is $O(\| \CCC \|)$, when 
$$m = \Theta((1/\gamma^2)\max\{(w(\XXX)+w(\CCC))^2,\log(T/\beta)\}).$$ Let $\EEE_0$ be the event that the above Lipschitz bound holds. 

\begin{lemma} \label{lem:JLzeroth}
For any $\beta > 0$, with probability at least $1 - \beta$, for each $t \in [T]$, $\theta^\priv_t$ generated by Algorithm~\ProjPrivIncReg satisfies:
\begin{align*}
\LLL_\proj(\theta^\priv_t;\Gamma_t;\Phi) - \LLL_\proj(\hat{\theta}_t;\Gamma_t;\Phi) = O \left (\frac{\sqrt{m} \log^{3/2} T \sqrt{\log(1/\delta)} \| \CCC \|^2}{\eps} \right )
\end{align*}
where $\hat{\theta}_t \in \mbox{argmin}_{\theta \in \CCC}\, \sum_{i=1}^t (y_i - \langle \x_i,\theta \rangle)^2$.
\end{lemma}
\begin{proof}
Let us condition on event $\EEE_0$ to hold. By definition,
\begin{align*}\min_{\vartheta \in \Phi\CCC}\, \sum_{i=1}^t \left (y_i - \frac{\| \x_i \|}{\| \Phi \x_i \|} \langle \Phi \x_i, \vartheta \rangle \right )^2 
\equiv \min_{\theta \in \CCC}\,\sum_{i=1}^t \left (y_i - \frac{\| \x_i \|}{\| \Phi \x_i \|} \langle \Phi \x_i, \Phi\theta \rangle \right )^2.
\end{align*}

Since the inputs are $m$-dimensional and $\vartheta^\priv_t = \Phi\theta^\priv_t$, using an analysis similar to Theorem~\ref{thm:thm2} gives that, with probability at least $1-\beta$, for each $t \in [T]$,
\begin{multline*}
\sum_{i=1}^t \left (y_i - \frac{\| \x_i \|}{\| \Phi \x_i \|} \langle \Phi \x_i, \Phi\theta^\priv_t \rangle \right )^2  - \min_{\theta \in \CCC}\,\sum_{i=1}^t \left (y_i - \frac{\| \x_i \|}{\| \Phi \x_i \|} \langle \Phi \x_i, \Phi\theta \rangle \right )^2 \\
=  O \left ( \frac{\log^{3/2} T \sqrt{\log(1/\delta)} \| \CCC \|^2 (\sqrt{m}+\sqrt{\log(T/\beta)})}{\eps}\right ).
\end{multline*}
In other words, with probability at least $1-\beta$, for each $t \in [T]$,
\begin{align*} 
\LLL_\proj(\theta^\priv_t;\Gamma_t;\Phi) -  \min_{\theta \in \CCC}\, \LLL_\proj(\theta;\Gamma_t;\Phi) =  O \left ( \frac{\log^{3/2} T \sqrt{\log(1/\delta)} \| \CCC \|^2 (\sqrt{m}+\sqrt{\log(T/\beta)})}{\eps}\right ).
\end{align*}
By noting that $\min_{\theta \in \CCC}\, \LLL_\proj(\theta;\Gamma_t;\Phi) \leq \LLL_\proj(\hat{\theta}_t;\Gamma_t;\Phi)$ and removing the conditioning on $\EEE_0$ (by adjusting $\beta$), completes the proof.
\end{proof}

Using properties of random projections, we now bound $\LLL_\proj(\hat{\theta}_t;\Gamma_t;\Phi)$ in terms of $\LLL(\hat{\theta}_t;\Gamma_t)$.
\begin{lemma} \label{lem:JLfirst}
Let $\Phi$ be a random matrix as defined in Theorem~\ref{thm:gordon} with $m = \Theta((1/\gamma^2)\max\{(w(\XXX)+w(\CCC))^2,\log(T/\beta)\})$, and let $\beta > 0$. Then with probability at least $1-\beta$, for each $t \in [T]$,
\begin{align*}
\LLL_\proj(\hat{\theta}_t;\Gamma_t;\Phi) \leq \LLL(\hat{\theta}_t;\Gamma_t) + 4 \gamma^{2} \| \CCC \|^{2} t + 2\gamma \| \CCC \| \sqrt{t \LLL(\hat{\theta}_t;\Gamma_t)}+ 2\sqrt{2}\gamma^{3/2} \| \CCC \|^{3/2}t^{3/4}\LLL(\hat{\theta}_t;\Gamma_t)^{1/4}.
\end{align*}
\end{lemma}
\begin{proof}
Fix a $t \in [T]$. From Theorem~\ref{thm:gordon}, for the chosen value of $m$, with probability at least $1-\beta$,
\begin{align} \label{eqn:JLfirst}
\LLL(\hat{\theta}_t;(\tilde{\x}_1,y_1),\dots,(\tilde{\x}_t,y_t)) & =\sum_{i=1}^t (y_i - \langle  \tilde{\x}_i, \hat{\theta}_t \rangle)^2 \nonumber \\
& = \sum_{i=1}^t  \left (y_i - \frac{\| \x_i\|}{\| \Phi \x_i \|} \langle  \x_i, \hat{\theta}_t \rangle \right )^2  \nonumber\\
& \leq \sum_{i=1}^t (|y_i - \langle  \x_i, \hat{\theta}_t \rangle| + \gamma \| \hat{\theta}_t \|)^2  \leq \sum_{i=1}^t (|y_i - \langle  \x_i, \hat{\theta}_t \rangle| + \gamma \| \CCC \|)^2 \nonumber\\
& \leq \LLL(\hat{\theta}_t;\Gamma_t) + \gamma^{2} \| \CCC \|^{2} t + 2\gamma \| \CCC \| \sum_{i=1}^t |y_i - \langle \x_i, \hat{\theta}_t \rangle| \nonumber\\
& \leq \LLL(\hat{\theta}_t;\Gamma_t) + \gamma^{2} \| \CCC \|^{2} t + 2\gamma \| \CCC \| \sqrt{t \LLL(\hat{\theta}_t;\Gamma_t)}.
\end{align}
Consider $\LLL_\proj(\hat{\theta}_t;\Gamma_t;\Phi)$. From Corollary~\ref{cor:gordon}, for the chosen value of $m$, with probability at least $1-\beta$,
\begin{align*}
\LLL_\proj(\hat{\theta}_t;\Gamma_t;\Phi) & = \sum_{i=1}^t (y_i - \langle  \Phi \tilde{\x}_i, \Phi \hat{\theta}_t \rangle )^2  \\
& \leq \sum_{i=1}^t (|y_i - \langle \tilde{\x}_i, \hat{\theta}_t \rangle| + \gamma \| \hat{\theta}_t \|)^2 \\
& \leq \sum_{i=1}^t (|y_i - \langle \tilde{\x}_i, \hat{\theta}_t \rangle| + \gamma \| \CCC \|)^2 \\
& = \LLL(\hat{\theta}_t;(\tilde{\x}_1,y_1),\dots,(\tilde{\x}_t,y_t)) + \gamma^{2} \| \CCC \|^{2} t + 2\gamma \| \CCC \| \sum_{i=1}^t |y_i - \langle \tilde{\x}_i, \hat{\theta}_t \rangle| \\
& \leq \LLL(\hat{\theta}_t;(\tilde{\x}_1,y_1),\dots,(\tilde{\x}_t,y_t)) + \gamma^{2} \| \CCC \|^{2} t + 2\gamma \| \CCC \| \sqrt{t \LLL(\hat{\theta}_t;(\tilde{\x}_1,y_1),\dots,(\tilde{\x}_t,y_t))},
\end{align*}
where the first inequality is by application of Corollary~\ref{cor:gordon}. Substituting the result from~\eqref{eqn:JLfirst} and taking a union bound over all $t \in [T]$ (i.e., replacing $\beta$ by $\beta/T$), completes the proof.
\end{proof}

The next step is to lower bound $\LLL_\proj(\theta^\priv_t;\Gamma_t;\Phi)$ in terms of $\LLL(\theta^\priv_t;\Gamma_t)$.
\begin{lemma} \label{lem:JLfifth}
For any $\beta > 0$, with probability at least $1 - \beta$, for each $t \in [T]$, $\theta^\priv_t$ generated by Algorithm~\ProjPrivIncReg satisfies:
\begin{multline*}
\LLL(\theta^\priv_t;\Gamma_t) \leq \LLL_\proj(\theta^\priv_t;\Gamma_t;\Phi) + 2 \gamma \| \CCC \| \sqrt{T \LLL(\theta^\priv_t;\Gamma_t)} \\
+  2\sqrt{2}\gamma^{3/2}\| \CCC \|^{3/2}T^{3/4}\LLL(\theta^\priv_t;\Gamma_t)^{1/4} + 4\gamma^2 \|\CCC\|^2 T. 
\end{multline*}
\end{lemma} 
\begin{proof}
Fix a $t \in [T]$. For the chosen value of $m$, 
$$\Pr\left[ \left | \| \Phi \x_i \| - \| \x_i \| \right | \geq \gamma \| \x_i \|~\mbox{\rm for all}~i\in[t]\right] \leq \beta.$$
Similarly for the chosen value of $m$ by using~\eqref{eqn:xthetprod},
$$\Pr\left [\left | \langle \Phi \x_i, \Phi \theta^\priv_t \rangle - \langle \x_i,\theta^\priv_t \rangle  \right | \geq \gamma \| \x_i \| \| \theta^\priv_t\|~\mbox{\rm for all}~i\in[t]\right ] \leq \beta.$$
Using arguments as in Lemma~\ref{lem:JLfirst}, but by focusing on the lower bounds, we get that with probability at least $1-\beta$,
\begin{multline*}
\LLL_\proj(\theta^\priv_t;\Gamma_t;\Phi) \geq \LLL(\theta^\priv_t;\Gamma_t) - 2 \gamma \| \CCC \| \sqrt{T \LLL(\theta^\priv_t;\Gamma_t)} \\
-  2\sqrt{2}\gamma^{3/2}\| \CCC \|^{3/2}T^{3/4}\LLL(\theta^\priv_t;\Gamma_t)^{1/4} - 4\gamma^2 \|\CCC\|^2 T.
\end{multline*}
Taking a union bound over all $t \in [T]$ completes the proof.
\end{proof}

Putting together Lemmas~\ref{lem:JLzeroth},~\ref{lem:JLfirst}, and~\ref{lem:JLfifth} and simple arithmetic manipulation gives: that with probability at least $1-\beta$, for each $t \in [T]$,
\begin{multline} \label{eqn:reexp}
\LLL(\theta^\priv_t;\Gamma_t) -  \LLL(\hat{\theta}_t;\Gamma_t) = O \left (\frac{\sqrt{m}\log^{3/2} T \sqrt{\log(1/\delta)} \| \CCC \|^2}{\eps} \right ) + 8\gamma^{2} \| \CCC \|^{2} T + 2 \gamma \| \CCC \| \sqrt{T \LLL(\hat{\theta}_t;\Gamma_t)}  \\
+ 2 \gamma \| \CCC \| \sqrt{T \LLL(\theta^\priv_t;\Gamma_t)} + 2\sqrt{2}\gamma^{3/2}\| \CCC \|^{3/2}T^{3/4}\LLL(\theta^\priv_t;\Gamma_t)^{1/4} + 2\sqrt{2}\gamma^{3/2}\| \CCC\|^{3/2}T^{3/4} \LLL(\hat{\theta}_t;\Gamma_t)^{1/4} 
.\end{multline}
To simplify~\eqref{eqn:reexp} in terms of its dependence on $\LLL(\theta^\priv_t;\Gamma_t)$, start by noting that it is of the form $p - a\sqrt{p} - a^{3/2}\sqrt[4]{p} - b - 2a^2 \geq 0$, where
\begin{flalign*}
& p  =  \LLL(\theta^\priv_t;\Gamma_t),\\
& a  =  2\gamma \| \CCC \| \sqrt{T}, \mbox{and} \\
& b  =   \LLL(\hat{\theta}_t;\Gamma_t)  + O \left ( \frac{\sqrt{m} \log^{3/2} T \sqrt{\log(1/\delta)} \| \CCC \|^2}{\eps} \right ) + 2 \gamma \| \CCC \| \sqrt{T \LLL(\hat{\theta}_t;\Gamma_t)} + 2\sqrt{2}\gamma^{\frac 3 2}\| \CCC\|^{\frac 3 2}T^{\frac 3 4} \LLL(\hat{\theta}_t;\Gamma_t)^{\frac 1 4}.
\end{flalign*}
Solving for $p$ from $p - a\sqrt{p} - a^{3/2}p^{1/4} - b - 2a^2 = 0$ (we use WolframAlpha solver~\cite{wolfram}), and using that to simplify~\eqref{eqn:reexp} yields, 
that with probability at least $1-\beta$, for each $t \in [T]$,
\begin{multline} \label{eqn:finexp}
\LLL(\theta^\priv_t;\Gamma_t) -  \LLL(\hat{\theta}_t;\Gamma_t)  = O( \frac{\sqrt{m} \log^{3/2} T \sqrt{\log(1/\delta)} \| \CCC \|^2}{\eps} \\
+ \gamma^{2} \| \CCC \|^{2} T + \gamma \| \CCC \| \sqrt{T \LLL(\hat{\theta}_t;\Gamma_t)} + \gamma^{3/2}\| \CCC\|^{3/2}T^{3/4} \LLL(\hat{\theta}_t;\Gamma_t)^{1/4} ).
\end{multline}

\begin{theorem} \label{thm:thm6}
Algorithm~\ProjPrivIncReg is $(\eps,\delta)$-differentially private with respect to a single datapoint change in the stream~$\Gamma$. For any $\beta > 0$, with probability at least $1-\beta$, for each $t \in [T]$, $\theta^\priv_t$ generated by Algorithm~\ProjPrivIncReg satisfies: 
\begin{multline*}
\LLL(\theta^\priv_t;\Gamma_t) -  \LLL(\hat{\theta}_t;\Gamma_t) = O (\frac{T^{\frac 1 3} W^{\frac 2 3} \log^{2} T \| \CCC \|^2  \sqrt{\log(1/\delta)\log(1/\beta)}}{\eps} \\
+ T^{\frac 1 6}W^{\frac 1 3} \| \CCC \|\sqrt{\LLL(\hat{\theta}_t;\Gamma_t)} + T^{\frac 1 4}W^{\frac 1 2} \| \CCC \|^{\frac 3 2} \sqrt[4]{\LLL(\hat{\theta}_t;\Gamma_t)}),
\end{multline*}
where $\hat{\theta}_t \in \min_{\theta \in \CCC}\, \LLL(\theta;\Gamma_t)$ and $W=w(\XXX)+w(\CCC)$.
\end{theorem}
\begin{proof}
The $(\eps,\delta)$-differential privacy follows from the established global sensitivity bound.

For the utility analysis, we start from~\eqref{eqn:finexp}, and substitute $\gamma = (w(\XXX)+w(\CCC))^{1/3}/T^{1/3}$ to get the claimed bound. The value of $\gamma$ is picked to balance the various opposing factors.
\end{proof}

\begin{remark}
Since $\LLL$ is an non-decreasing function in $t$, in the above theorem $\LLL(\hat{\theta}_t;\Gamma_t)$ in the right-hand side could be replaced by $\LLL(\hat{\theta}_T;\Gamma_T)$ (defined as $\OPT$ in the introduction).
\end{remark}

\subsection{Discussion about Theorem~\ref{thm:thm6}} \label{sec:inst}
We start this discussion by mentioning a few instantiations of Theorem~\ref{thm:thm6}. Let $\OPT = \LLL(\hat{\theta}_T;\Gamma_T)$ (remember, that $\OPT \leq T$). For simplicity, below we ignore dependence on the privacy and confidence parameters.

For arbitrary $\XXX$ and $\CCC$, with just an $L_2$-diameter assumption as in Theorem~\ref{thm:thm2}, $W=w(\XXX)+w(\CCC)=O(\sqrt{d})$, and therefore the excess risk bound provided by Theorem~\ref{thm:thm6} (accounting for the trivial excess risk bound of $T$) is $\tilde{O}(\min\{T^{1/3} d^{1/3} + T^{1/6} d^{1/6} \sqrt{\OPT} + T^{1/4} d^{1/4} \sqrt[4]{\OPT},T\})$, which is worse than the $\tilde{O}(\min\{\sqrt{d},T\})$ excess risk bound provided by Theorem~\ref{thm:thm2}. However, as we discuss below, for many interesting high-dimensional problems, one could get substantially better bounds using Theorem~\ref{thm:thm6}. This happens when $W$ is ``small''.

In many practical regression instances, the input data is high-dimensional and sparse~\cite{nelson2013osnap,halko2011finding,zhou2008compressed} which leads to a small $w(\XXX)$. For example, if the $\x_i$'s are $k$-sparse, then $w(\XXX) = O(\sqrt{k \log(d/k)})$. Another common scenario is to have the $\x_i$'s come from a bounded $L_1$-diameter ball, in which case $w(\XXX)= O(\sqrt{\log d})$. Now under any of these assumptions, let us look at different choices of constraint spaces ($\CCC$) that have a small Gaussian width.

\begin{CompactItemize}
\item If $\CCC=\mathrm{conv}\{\a_1,\dots,\a_l\}$ be convex hull of vectors $\a_i \in \R^d$, such that for all $i \in [d], \| \a_i \| \leq c$ with $c \in \R^+$, then $w(\CCC) = O(c \sqrt{\log l})$. A popular subcase of the above is the {\em cross-polytope} $B_1^d$ (unit $L_1$-ball). For example, the popular Lasso formulation~\cite{tibshirani1996regression} used for high-dimensional linear regression is:
$$\hat{\theta}_t \in \mbox{argmin}_{\theta \in c B_1^d}\, \sum_{i=1}^t (y_i - \langle \x_i, \theta \rangle)^2.$$
Another popular subcase is the probability simplex where, $\CCC = \{ \theta \in \R^d \,:\, \sum_{i} \theta_i = 1, \forall i \in [d], \theta_i \geq 0 \}$.

\item Group/Block $L_1$-norm is another prominent sparsity inducing norm used in many applications~\cite{bach2012optimization}. For a vector $\theta \in \R^d$, and a parameter $k$, this norm is defined as:\!\footnote{There are generalizations of this norm that can handle different group (block) sizes.}
$$\|\theta\|_{k,L_{1,2}} = \sum_{i=1}^{\lceil d/k \rceil} \sqrt{\sum_{j=(i-1)k+1}^{\min\{ik, d\}} |\theta_j|^2}.$$
If $\CCC$ denotes the convex set centered with radius one with respect to $\|\cdot\|_{k,L_{1,2}}$-norm then the Gaussian width of $\CCC$ is $O(\sqrt{k\log(d/k)})$~\cite{talwar2014private}.

\item $L_p$-balls ($1 < p < 2$) are another popular choice of constraint space~\cite{rahimi2013norm}. The regression problem in this case is defined as:
$$\hat{\theta}_t \in \mbox{argmin}_{\theta \in c B_p^d}\, \sum_{i=1}^t (y_i - \langle \x_i, \theta \rangle)^2,$$
and $w(c B_p^d) = O(cd^{1-1/p})$.
\end{CompactItemize}

As the reader may notice, in all of the above problem settings, $W$ is much smaller than $\sqrt{d}$. As a comparison to the bound in Theorem~\ref{thm:thm2}, if  $W =O(\polylog(d))$, then Theorem~\ref{thm:thm6} yields an excess risk bound of $\tilde{O}(T^{1/3} + T^{1/6}\sqrt{\OPT} + T^{1/4} \sqrt[4]{\OPT})$. This is significantly better than the $\tilde{O}(\sqrt{d})$ risk bound from Theorem~\ref{thm:thm2} for many setting of $T,d$, and $\OPT$, e.g., if $d \gg T^{4/3}$. 


One could also compare the result from Theorem~\ref{thm:thm6} to the bound obtained by applying the differentially private ERM algorithm of Talwar~\emph{et al.}\ \cite{talwar2015nearly} in the generic mechanism (Theorem~\ref{thm:thm1}, Part~\ref{res:incerm3}). It is hard to do a precise comparison because of the dependence on different parameters in these bounds. In general, when $\OPT$ is not very big (say $\ll T^{2/3}$) and $W=O(\polylog(d))$ then the excess risk bound from Theorem~\ref{thm:thm6} is significantly better than $\tilde{O}(\sqrt{T})$ risk bound obtained in Theorem~\ref{thm:thm1}, Part~\ref{res:incerm3}.


\smallskip
\noindent\textbf{Extension to a case where not all inputs are drawn from a domain with small Gaussian Width.} The previous analysis assumes that $w(\XXX)$ is small (all inputs are drawn from a domain with small Gaussian Width). We now show that the techniques and results in the previous section extend to a more robust setting, where not all inputs are assumed to come from a domain with small Gaussian width. 

In particular, we assume that there exists a set $\GGG \subseteq \XXX$, such that $w(\GGG)$ is small, and only some inputs in the stream come from $\GGG$ (e.g., only a fraction of the covariates could be sparse). We also assume that the algorithm has access to an oracle which given a point $\x \in \XXX$, returns whether $\x \in \GGG$ or not. The goal of the algorithm is to perform private incremental linear regression on inputs from $\GGG$. In a non-private world, this is trivial as the algorithm can simply ignore when $\x_t$ is not in $\GGG$, but this operation is not private. However, a simple change to Algorithm~\ProjPrivIncReg can handle this scenario without a breach in the privacy. The idea is to check whether $\x_t \in \GGG$, if so $(\x_t,y_t)$ is used as currently in Algorithm~\ProjPrivIncReg. Otherwise, $(\x_t,y_t)$ is replaced by $(\mathbf{0},0)$ before invoking Algorithm~\TreeMech (in Steps~\ref{step:first} and~\ref{step:second} of Algorithm~\ProjPrivIncReg). With this change, the resulting algorithm is $(\eps,\delta)$-differentially private, and with probability at least $1-\beta$, for each $t \in [T]$, its output $\theta^\priv_t$ will satisfy:
\begin{multline*}
\sum_{\x_i \in \GGG, i \in [t]}(y_i - \langle \x_i, \theta^\priv_t \rangle)^2 -  \sum_{\x_i \in \GGG, i \in [t]}(y_i - \langle \x_i, \hat{\theta}_t \rangle)^2 
 = O(\frac{T^{\frac 1 3} W^{\frac 2 3} \log^{2} T \| \CCC \|^2  \sqrt{\log(1/\delta) \log(1/\beta)}}{\eps} \\ 
+ T^{\frac 1 6}W^{\frac 1 3} \| \CCC \|\sqrt{\LLL(\hat{\theta}_t;\Gamma_t)} + T^{\frac 1 4}W^{\frac 1 2} \| \CCC \|^{\frac 3 2} \sqrt[4]{\LLL(\hat{\theta}_t;\Gamma_t)}),
\end{multline*}
where $\hat{\theta}_t \in \min_{\theta \in \CCC}\, \sum_{\x_i \in \GGG, i \in [t]}(y_i - \langle \x_i, \theta \rangle)^2$ and $W=w(\GGG)+w(\CCC)$.

\appendix

\section{Additional Preliminaries} \label{app:convex}
We start by reviewing some standard definitions in convex optimization. In our setting, for a loss function $\jmath(\theta;\z)$, all the following properties (such as convexity, Lipschitzness, strong convexity) are defined with respect to the first argument $\theta$.

In the following, we use the notation $\nabla \jmath(\theta;\z)$ to denote the gradient (if it exists) or any subgradient of function $\jmath(\cdot;\z)$ at $\theta$. 

\begin{definition} (Convex Functions) A loss function $\jmath \,:\, \CCC \times \ZZZ \rightarrow \R$ is convex with respect to $\theta$ over the domain $\CCC$ if for all $\z \in \ZZZ$ the following inequality holds: $\jmath(\lambda\theta_a + (1-\lambda)\theta_b;\z) \leq \lambda\cdot \jmath(\theta_a;\z)+(1-\lambda) \cdot \jmath(\theta_b;\z)$ for all $\theta_a,\theta_b \in \CCC$ and $\lambda \in [0,1]$. For a continuously differentiable $\jmath$, this inequality can be equivalently replaced with: $\jmath(\theta_b;\z) \geq  \jmath(\theta_a;\z) + \langle \nabla \jmath(\theta_a;\z), \theta_b - \theta_a \rangle$ for all $\theta_a,\theta_b \in \CCC$, where $\nabla \jmath(\theta_a;\z)$ is the gradient of $\jmath(\cdot;\z)$ at $\theta_a$.
\end{definition}


\begin{definition} (Lipschitz Functions) \label{defn:Lipsfunc}
A loss function $\jmath \,:\, \CCC \times \ZZZ \rightarrow \R$ is $L$-Lipschitz with respect to $\theta$ over the domain $\CCC$, if for all $\z \in \ZZZ$ and $\theta_a,\theta_b \in \CCC$, we have $|\jmath(\theta_a;\z) - \jmath(\theta_b;\z) | \leq L \| \theta_a - \theta_b \|$. If $\jmath$ is a convex function, then $\jmath$ is $L$-Lipschitz iff for all $\theta \in \CCC$ and subgradients $\g$ of $\jmath$ at $\theta$ we have $\| \g \| \leq L$.
\end{definition}

\begin{definition} (Strongly Convex Functions) \label{defn:strongconvex}
A loss function $\jmath \,:\, \CCC \times \ZZZ \rightarrow \R$ is $\nu$-strongly convex if for all $\z \in \ZZZ$ and $\theta_a, \theta_b \in \CCC$, all subgradients $\g$ of $\jmath(\theta_a;\z)$, we have $\jmath(\theta_b;\z) \geq \jmath(\theta_a;\z) + \langle \g, \theta_b-\theta_a \rangle + (\nu/2) \| \theta_b - \theta_a \|^2$  (i.e., $\jmath$ is bounded below by a quadratic function tangent at $\theta_a$).
\end{definition}

\smallskip
\noindent\textbf{Gaussian Norm Bounds.} Let $\NNN(0,\sigma^2)$ denote the Gaussian (normal) distribution with mean $0$ and variance $\sigma^2$. We use the following standard result on the spectral norm (largest singular value) of an i.i.d.\ Gaussian random matrix throughout this paper.
\begin{proposition} \label{prop:GaussConc}
Let $A$ be an $N \times n$ matrix whose entries are independent standard normal random variables. Then for every $t > 0$, with probability at least $1 - 2 \exp(-t^2/2)$ one has, $\| A \|  = O(\sqrt{N}+ \sqrt{n} + t)$. In particular, with probability at least $1 - \beta $, $\| A \|  = O(\sqrt{N}+ \sqrt{n} + \sqrt{\log (1/\beta)})$.
\end{proposition}

\subsection{Background on Linear Regression} \label{app:linearregressionbackground}
Linear regression is a statistical method used to create a linear model. It attempts to model the relationship between two variables (known as covariate-response pairs) by fitting a linear function to observed data. More formally, given $\y = X \theta^\star+ \w$, where $\y = (y_1,\dots,y_n) \in \R^n$ is a vector of observed responses, $X\in \R^{n \times d}$ is the covariate matrix in which $i$th row $\x_i^\top$ represents the covariates (features) for the $i$th observation, and $\w = (w_1,\dots,w_n)$ is a noise vector, the goal of linear regression is to estimate the unknown regression vector $\theta^\star$.

Assuming that the noise vector $\w$ follows a (sub)Gaussian distribution, estimating $\theta^\star$ amounts to solving the ordinary least-squares problem: 
$$\hat{\theta} \in \mbox{argmin}_{\theta \in \R^d}\, \| \y - X \theta \|^2 = \mbox{argmin}_{\theta \in \R^d}\, \sum_{i=1}^n (y_i - \langle \x_i,\theta \rangle)^2.$$
Typically, for additional guarantees such as sparsity and stability, constraints are added to the least-squares estimator. This leads to the constrained linear regression formulation:
\begin{align} \label{eqn:genlinearreg}
\mbox{Linear Regression:} \quad \hat{\theta} \in \mbox{argmin}_{\theta \in \CCC}\, \| \y - X \theta \|^2 = \mbox{argmin}_{\theta \in \CCC}\, \sum_{i=1}^n (y_i - \langle \x_i,\theta \rangle)^2,
\end{align}
for some convex set $\CCC \subseteq \R^d$. In this paper, we work with this formulation. Two well-known regression problems are obtained by choosing $\CCC$ as the $L_2$/$L_1$-ball:
\begin{eqnarray*}
&\mbox{Ridge Regression:} \quad \hat{\theta} \in \mbox{argmin}_{\theta \in c B_2^d}\, \sum_{i=1}^n (y_i - \langle \x_i,\theta \rangle)^2,& \nonumber \\
&\mbox{Lasso Regression:} \quad \hat{\theta} \in \mbox{argmin}_{\theta \in c B_1^d}\, \sum_{i=1}^n (y_i - \langle \x_i,\theta \rangle)^2,&
\end{eqnarray*}
where $c \in \R^+$.
Another popular example is the Elastic-net regression which combines the Lasso and Ridge regression. Note that, while in this paper we focus on the constrained formulation of regression, from duality and the KKT conditions, the constrained formulation is equivalent to a penalized (regularized) formulation. 

\smallskip
\noindent\textbf{In the Streaming Setting.} Let $\Gamma = (\x_1,y_1),\dots,(\x_T,y_T)$ denote a data stream. Let $\Gamma_t$ denote the stream prefix of $\Gamma$ of length $t$. Informally, the goal of {\em incremental linear regression} is to release at each timestep $t \in [T]$, $\theta_t \in \CCC$, that minimizes $\LLL(\theta;\Gamma_t) = \sum_{i=1}^t (y_i - \langle \x_i,\theta\rangle)^2$. 

\begin{definition} [Adaptation of Definition~\ref{defn:utilerm} for Incremental Linear Regression] \label{defn:utililr}
A randomized streaming algorithm is $(\alpha,\beta)$-{\em estimator} for incremental linear regression, if with probability at least $1 -\beta$ over the coin flips of the algorithm, for each $t \in [T]$, after processing a prefix of the stream of length $t$, it generates an output $\theta_t \in \CCC$ that satisfies the following bound on excess (empirical) risk:
$$\sum_{i=1}^t (y_i - \langle \x_i,\theta_t\rangle)^2 - \left ( \min_{\theta \in \CCC}\, \sum_{i=1}^t (y_i - \langle \x_i,\theta\rangle)^2 \right ) \leq \alpha.$$ 
\end{definition}


\subsection{Background on Differential Privacy} \label{app:dp}
In this section, we review some basic constructions in differential privacy. The literature on differential privacy is now rich with tools for constructing differentially private analyses, and we refer the reader to a survey by Dwork and Roth~\cite{dwork2013algorithmic} for a comprehensive review of developments there.

One of the most basic technique, for achieving differential privacy is by adding noise to the outcome of a computed function where the noise magnitude is scaled to the {\em (global) sensitivity} of the function defined as:
\begin{definition}[Sensitivity] \label{defn:l2sens}
Let $f$ be a function mapping streams $\Gamma \in \ZZZ^\ast$ to $\R^d$. The $L_2$-sensitivity $\Delta_2$ of $f$ is the maximum of $\|f(\Gamma)-f(\Gamma')\|$ over neighboring streams $\Gamma,\Gamma'$. 
\end{definition}

\begin{theorem}[Framework of Global Sensitivity~\cite{DMNS06}]\label{thm:sens} Let $f: \ZZZ^\ast \rightarrow \R^d$ be a function with $L_2$-sensitivity $\Delta_2$. The algorithm that on an input $\Gamma$ outputs $f(\Gamma) + Y \mbox{\rm where}\, Y\sim \NNN(0,(2 \Delta_2^2 \ln(2/\delta))/\eps^2)^d$ is $(\epsilon,\delta)$-differentially private. 
\end{theorem}

Composition theorems for differential privacy allow a modular design of privacy preserving algorithms based on algorithms for simpler sub tasks:

\begin{theorem}[\cite{DKMMN06}]\label{thm:composition1}
A mechanism that permits $k$ adaptive interactions with mechanisms that preserves $(\epsilon,\delta)$-differential privacy (and does not access the database otherwise) ensures $(k\epsilon, k\delta)$-differential privacy.
\end{theorem}

A stronger composition is also possible as shown by Dwork~\emph{et al.} \cite{DRV10}.

\begin{theorem}[\cite{DRV10}]\label{thm:composition2}
Let $\epsilon,\delta,\delta^\ast>0$ and $\epsilon \leq 1$. A mechanism that permits $k$ adaptive interactions with mechanisms that preserves $(\epsilon,\delta)$-differential privacy 
ensures $(\epsilon\sqrt{2k\ln(1/\delta^\ast)}+ 2 k\epsilon^2, k\delta+\delta^\ast)$-differential privacy.
\end{theorem}


\section{Noisy Projected Gradient Descent} \label{app:noisyproj}
In this section, we investigate the convergence rate of noisy projected gradient descent. Williams and McSherry first investigated gradient descent with noisy updates for probabilistic inference~\cite{williams2010probabilistic}; noisy stochastic variants of gradient descent have also been studied by~\cite{DBLP:journals/jmlr/JainKT12,duchi2013local,song2013stochastic,bassily2014differentially} in various private convex optimization settings. Other convex optimization techniques such as mirror descent~\cite{duchi2014privacy,talwar2014private} and Frank-Wolfe scheme~\cite{talwar2014private} have also been considered for designing private ERM algorithms. 

For completeness, in this section, we present an analysis of the projected gradient descent procedure that operates with only access to a private gradient function (Definition~\ref{defn:privgradfunc}). The analysis relies on standard ideas in convex optimization literature.

Consider the following constrained optimization problem,
\begin{equation} \label{eqn:npg} \min_{\theta \in \CCC}\, f(\theta) \mbox{ where } \CCC \subseteq \R^d, \end{equation}
where $f: \R^d \rightarrow \R$ is a convex function and $\CCC$ is some non-empty closed convex set. Define projection of $\theta \in \R^d$ onto a convex set $\CCC$ as:
$$P_{\CCC}(\theta) = \mbox{argmin}_{\z \in \CCC}\, \| \theta - \z \|^2.$$

Projected gradient descent algorithm uses the following update to solve~\eqref{eqn:npg}
\begin{align} \label{eqn:projgrad} 
\ProjGrad(\CCC,r): \mbox{Initialize } \theta_1 \in \CCC, \mbox{Repeat $r$ times: }  
\theta_{k+1} = P_{\CCC}(\theta_k - \eta_k \nabla f(\theta_k)), \mbox{ Output }  \bar{\theta} =  \frac{1}{r}\sum_{i=1}^{r} \theta_i, 
\end{align}
for some stepsize $\eta_k$. Here $P_{\CCC}(\theta)$ defines the projection of $\theta$ onto $\CCC$. 

Let $g : \R^d \rightarrow \R^d$ be an $(\alpha,\beta)$-approximation of the true gradient of $f$ (as in Definition~\ref{defn:privgradfunc}),  
$$\Pr\left[\max_{\theta \in \CCC}\| g(\theta) - \nabla f(\theta)  \| > \alpha\right]\leq\beta.$$
The noisy projected gradient descent is a simple modification of the projected gradient descent algorithm~\eqref{eqn:projgrad}, where instead of the true gradient, a noisy gradient is used. In other words, noisy projected gradient descent takes the form,
\begin{align}
\NoisyProjGrad(\CCC,g,r): \mbox{Initialize } \theta_1 \in \CCC, \mbox{Repeat $r$ times: } 
\theta_{k+1} = P_{\CCC}(\theta_k - \eta_k g(\theta_k)), \mbox{ Output }  \bar{\theta} = \frac{1}{r}\sum_{i=1}^{r} \theta_i.
\end{align}

The following proposition analyzes the convergence of the above $\NoisyProjGrad$ procedure assuming $g$ is an $(\alpha,\beta)$-approximation of the true gradient of $f$. Note that the proposition holds, even if $f$ is not differentiable, in which case, $\nabla f(\theta)$ represents any subgradient of $f$ at $\theta$.

\begin{proposition} \label{prop:noisyprojgrad}
Suppose $f$ is convex and is $L$-Lipschitz. Let $\| \CCC \|$ be the diameter of $\CCC$. Let $g : \R^d \rightarrow \R^d$ be an $(\alpha,\beta)$-approximation of the true gradient of $f$. Then after $r$ steps of Algorithm~\NoisyProjGrad, starting with any $\theta_{0} \in \CCC$ and constant stepsize of $\eta_k = \frac{\| \CCC \|}{\sqrt{r}(\alpha+L)}$, with probability at least $1-r\beta$, 
$$f(\bar{\theta}) - f(\theta^\ast) \leq \frac{(\alpha+L) \| \CCC \|}{\sqrt{r}} + \alpha\| \CCC \|.$$ 
\end{proposition}
\begin{proof}
Let $\theta^\ast \in \CCC$ be the optimal solution of~\eqref{eqn:npg}. Let $g(\theta_k) = \nabla f(\theta_k) + e(\theta_k)$. By Definition~\ref{defn:Lipsfunc}, $L \geq  \max_{\theta \in \CCC} \| \nabla f(\theta) \|$.
\begin{align*}
& \| x_{k+1} - \theta^{\ast} \|^2  = \| P_{\CCC}(\theta_k -  \eta_k g(\theta_k)) - P_{\CCC}(\theta^\ast) \|^2 \\
& \leq \| \theta_k -  \eta_k g(\theta_k) - \theta^\ast  \|^2 \\
& \leq \| \theta_k -  \theta^\ast \|^2 + 2 \eta_k \langle g(\theta_k), \theta^\ast - \theta_k \rangle + \eta_k^2 \| g(\theta_k) \|^2 \\
& = \| \theta_k -  \theta^\ast \|^2 + 2 \eta_k  \langle \nabla f(\theta_k) + e(\theta_k), \theta^\ast - \theta_k \rangle  + \eta_k^2 \| \nabla f(\theta_k) + e(\theta_k) \|^2  \\
& \leq \| \theta_k -  \theta^\ast \|^2 + 2 \eta_k  \langle \nabla f(\theta_k), \theta^\ast - \theta_k \rangle +  2\eta_k\| e(\theta_k) \| \| \CCC \|  + \eta_k^2 \| \nabla f(\theta_k) + e(\theta_k) \|^2 \\
& \leq \| \theta_k -  \theta^\ast \|^2 + 2 \eta_k (f(\theta^\ast) - f(\theta_k)) + 2 \eta_k \| e(\theta_k) \| \| \CCC \|  + \eta_k^2 (\| \nabla f(\theta_k) \|^2  + \| e(\theta_k) \|^2 + 2 \| \nabla f(\theta_k) \|  \| e(\theta_k) \|) \\
& \leq \| \theta_k -  \theta^\ast \|^2 + 2 \eta_k (f(\theta^\ast) - f(\theta_k)) + 2 \eta_k\| e(\theta_k) \| \| \CCC \|  + \eta_k^2 ( L^2 + \| e(\theta_k) \|^2 + 2L \| e(\theta_k) \|).
\end{align*}
The first inequality follows because projection cannot increase distances (projection operator is {\em contractive}). 

By the assumption on $g(\theta_k)$, with probability at least $1-\beta$, $\| e(\theta_k) \| \leq \alpha$. Hence, with probability at least $1-\beta$,
\begin{align*}
\| \theta_{k+1} - \theta^{\ast} \|^2 \leq \| \theta_k -  \theta^\ast \|^2 + 2 \eta_k (f(\theta^\ast) - f(\theta_k)) + 2 \eta_k \alpha \| \CCC \| + \eta_k^2 ( L^2 + \alpha^2 + 2\alpha L).
\end{align*}
Summing the above expression and taking a union bound, yields that with probability at least $1-r\beta$,
\begin{align*}
0 \leq \| \theta_{r+1} - \theta^\ast \|^2 \leq \| \theta_1 - \theta^\ast \|^2 + 2 \sum_{k=1}^r \eta_k (f(\theta^\ast) - f(\theta_k))  + 2 r \eta_k \alpha  \| \CCC \| + r \eta_k^2 ( L^2 + \alpha^2 + 2\alpha L).
\end{align*}
Rearranging the above, and setting constant step size $\eta_k = \frac{\| \CCC \|}{\sqrt{r}(\alpha+L)}$, with probability at least $1-r\beta$,
$$\sum_{k=1}^r (f(\theta_k) - f(\theta^\ast)) \leq  \| \CCC \| \sqrt{r}(\alpha+L)  + r \alpha \| \CCC \|.$$

Now by Jensen's inequality, 
$$f(\bar{\theta}) = f \left ( \frac{1}{r} \sum_{k=1}^r \theta_k \right ) \leq \frac{1}{r} \sum_{k=1}^r f(\theta_k).$$
Therefore,
$$f(\bar{\theta}) - f(\theta^\ast) \leq \frac{1}{r} \sum_{k=1}^r (f(\theta_k) - f(\theta^\ast))  \leq \frac{(\alpha+L) \| \CCC \| }{\sqrt{r}}  + \alpha \| \CCC \|.$$
\end{proof}

\begin{corollary} \label{cor:noisyprojgrad}
By setting $r=\frac{(\alpha+L)^2 \| \CCC \|^2}{\zeta^2}$, in Proposition~\ref{prop:noisyprojgrad} gives that with probability at least $1-r\beta$, 
$$f(\bar{\theta}) - f(\theta^\ast) \leq \zeta + \alpha \| \CCC \|.$$
If $\alpha > 0$, then we can set $r = \left (1+\frac{L}{\alpha} \right)^2$ gives $f(\bar{\theta}) - f(\theta^\ast) \leq 2 \alpha\| \CCC \|$.
\end{corollary}

After constructing the private gradient function, evaluating the gradient function at any $\theta$ can be done without affecting the privacy budget, as this is just a post-processing of private outputs.

\section{Tree Mechanism for Continually Releasing Private Sums}\label{app:treemech}
Given a bit stream $b_1,\dots,b_T \in \{0,1\}$, the {\em private streaming counter} problem is to release at every timestep $t$, (an approximation to) $\sum_{i=1}^t b_i$ while satisfying differential privacy (Definition~\ref{defn:dp}). Chan~\emph{et al.}\ \cite{CSS11} and Dwork~\emph{et al.}\ \cite{DNPR10} proposed an elegant differentially private mechanism (referred to as the \emph{Tree Mechanism}) for this problem. We use this mechanism as a basic building block in our  private incremental regression algorithms.\!\footnote{Dwork \emph{et al.}\ \cite{DNRR15} have recently improved the bounds for the private streaming counter problem in the case where the bit stream is {\em sparse}, i.e., have many fewer $1$'s than $0$'s.}

For completeness, in Algorithm~\TreeMech, we present the entire \emph{Tree Mechanism} as applied to a set of vectors. Given a data stream $\Upsilon=\upsilon_1,\dots,\upsilon_T \in \ZZZ$, the algorithm releases at each timestep $t$, (an approximation to) the sum function $\sum_{i=1}^t \upsilon_i$, while satisfying differential privacy.


\begin{algorithm}[!t]
\DontPrintSemicolon
\caption{\TreeMech$(\eps,\delta,\Delta_2)$}
\KwIn{A stream $\Upsilon = \upsilon_1,\dots,\upsilon_T$, where each $\upsilon_t$ is from the domain $\ZZZ \subseteq \R^d$, and $\Delta_2 = \max_{\upsilon,\upsilon' \in \ZZZ} \| \upsilon - \upsilon' \|$}
\KwOut{A differentially private estimate of $\sum_{i=1}^t \upsilon_i$ at every timestep $t \in [T]$}
\For{all $t \in [T]$}{
Express $t = \sum_{j=0}^{\log t} 2^j \Bin_j(t)$ (where $\Bin_j(t)$ is the bit at the $j$th index in the binary representation of $t$)\;
$i \leftarrow \min_{0 \leq j \leq \log T}\, \{\Bin_j(t) \neq 0\}$\;
$\a_i \leftarrow \sum_{j < i} \a_j + \upsilon_t$\;
\For{$0 \leq j \leq i-1$}{
$\a_j \leftarrow \mathbf{0}$ and $\b_j \leftarrow \mathbf{0}$
}
$\b_i \leftarrow \a_i + \NNN \left (\mathbf{0},\frac{2 \log^2(T) \Delta_2^2 \ln(2/\delta) \mathbb{I}_d}{\eps^2} \right )$\;
$\s_t \leftarrow \sum_{j:\Bin_j(t) \neq 0} \b_j$\;
Return $\s_t$ 
}
\end{algorithm}
Algorithm~\TreeMech can be viewed as releasing partial sums of different ranges at each timestep $t$ and computing the final sum is simply a post-processing of the partial sums. At most $\log T$ partial sums are used for constructing each private sum. The following theorem follows by using the standard upper deviation inequality for Gaussian random variables (Proposition~\ref{prop:GaussConc}) in the analysis of \emph{Tree Mechanism} from~\cite{DNPR10,CSS11}. Another advantage of this mechanism is that it can be implemented with small memory, as only $O(\log t)$ partial sums are needed at any time $t$. 

\begin{proposition} \label{prop:counter}
Algorithm~\TreeMech is $(\eps,\delta)$-differentially private with respect to a single datapoint change in the stream $\Upsilon$. For any $\beta >0$ and $t \in [T]$, with probability at least $1-\beta$, $\s_t$ computed by Algorithm~\TreeMech satisfies:
$$ \left \|\s_t - \sum_{i=1}^t \upsilon_i \right \| = O \left (\frac{\Delta_2 (\sqrt{d} +\sqrt{\log \left (1/\beta \right )})  \log^{3/2} T\sqrt{\log \left ( 1/\delta \right ) }}{\eps} \right ),$$
where $\Delta_2$ is the $L_2$-sensitivity of the sum function from Definition~\ref{defn:l2sens}.
\end{proposition}





\end{document}